\documentclass[11pt,a4paper]{article}
\usepackage{a4wide}
\usepackage[utf8]{inputenc}
\usepackage{authblk}

\title{Maximizing Nash Social Welfare in 2-Value Instances}
\author[1]{Hannaneh Akrami}
\author[1]{Bhaskar Ray Chaudhury}
\author[2]{Martin Hoefer} 
\author[1]{Kurt Mehlhorn}
\author[2]{Marco Schmalhofer} 
\author[1]{Golnoosh Shahkarami}
\author[2]{Giovanna Varricchio}
\author[3]{Quentin Vermande}
\author[3]{Ernest van Wijland}

\affil[1]{Max Planck Institute for Informatics, Universit\"at des Saarlandes.}
\affil[ ]{
\texttt{\{hakrami, braycha, mehlhorn, gshahkar\}@mpi-inf.mpg.de}}
\affil[ ]{}
\affil[2]{Goethe University Frankfurt, Institute for Computer Science.}
\affil[ ]{
\texttt{\{mhoefer, varricchio, schmalhofer\}@em.uni-frankfurt.de}}
\affil[ ]{}
\affil[3]{\'Ecole Normale Sup\'erieure, Paris.}
\affil[ ]{ \texttt{\{quentin.vermande, ernest.van.wijland\}@ens.fr}}

\usepackage{amsmath}
\usepackage{amsfonts}
\usepackage{amssymb}
\usepackage{graphicx}
\usepackage{amsthm}
\usepackage{enumitem}
\usepackage{cleveref}
\usepackage{nameref}
\usepackage{bbm}
\usepackage{amsfonts}
\usepackage{verbatim}
\usepackage{mathtools}
\usepackage{mathrsfs}
\usepackage{multirow}
\usepackage{array}
\usepackage{subcaption}
\usepackage[linesnumbered,ruled,vlined]{algorithm2e}
\usepackage{cite}
\usepackage{breqn}
\usepackage{thm-restate}

\newcommand{\set}[1]{\{#1\}}
\newcommand{\transGraph}[2]{G_{#1\rightarrow#2}}

\newcommand{\modulus}[1]{\vert #1 \vert}
\newcommand{\goods}{\mathcal{G}}
\newcommand{\agents}{\mathcal{N}}

\newcommand{\NSW}{\text{NSW}}

\newcommand{\instance}{\mathcal{I}}

\newcommand{\NN}{\ensuremath{\mathbb{N}}}
\newcommand{\RR}{\ensuremath{\mathbb{R}}}
\newcommand{\QQ}{\ensuremath{\mathbb{Q}}}

\newcommand{\classNP}{\textsf{NP}}
\newcommand{\classAPX}{\textsf{APX}}

\newcommand{\calA}{{\cal A}}
\newcommand{\calG}{{\cal G}}
\newcommand{\calN}{{\cal N}}

\newcommand{\SB}{\textsf{SB}}
\newcommand{\BB}{\textsf{BB}}
\renewcommand{\SS}{\textsf{SS}}
\newcommand{\BS}{\textsf{BS}}

\newcommand{\Balance}{\textsc{Balance}}
\newcommand{\TwoValueApprox}{\textsc{TwoValueApprox}}

\newtheorem{proposition}{Proposition}
\newtheorem{theorem}{Theorem}
\newtheorem{example}{Example}
\newtheorem{lemma}{Lemma}
\newtheorem{corollary}{Corollary}
\newtheorem{definition}{Definition}
\newtheorem*{claim*}{Claim}

\crefname{algocf}{alg.}{algs.}
\Crefname{algocf}{Algorithm}{Algorithms}

\begin{document}
\maketitle

\begin{abstract}
	We consider the problem of maximizing the Nash social welfare when allocating a set $\calG$ of indivisible goods to a set $\calN$ of agents. We study instances, in which all agents have 2-value additive valuations: The value of every agent $i \in \calN$ for every good $j \in \calG$ is $v_{ij} \in \{p,q\}$, for $p,q \in \NN$, $p \le q$. Maybe surprisingly, we design an algorithm to compute an optimal allocation in polynomial time if $p$ divides $q$, i.e., when $p=1$ and $q \in \NN$ after appropriate scaling. The problem is \classNP-hard whenever $p$ and $q$ are coprime and $p \ge 3$. 
	
	In terms of approximation, we present positive and negative results for general $p$ and $q$. We show that our algorithm obtains an approximation ratio of at most 1.0345. Moreover, we prove that the problem is \classAPX-hard, with a lower bound of $1.000015$ achieved at $p/q = 4/5$.
\end{abstract}


\section{Introduction}

Fair division is an important area at the intersection of economics and computer science. While fair division with divisible goods is relatively well-understood in many contexts, the case of \emph{indivisible} goods is significantly more challenging. Recent work in fair division has started to examine extensions of standard fairness concepts such as envy-freeness to notions such as EF1 (envy-free up to one good)~\cite{LiptonMMS04} or EFX (envy-free up to any good)~\cite{CaragiannisKMPSW16}, most prominently in the case of non-negative, additive valuations of the agents. In this additive domain, notions of envy-freeness are closely related to the \emph{Nash social welfare (NSW)}, which is defined by the geometric mean of the valuations. An allocation maximizing the Nash social welfare is Pareto-optimal, satisfies EF1~\cite{CaragiannisKMPSW16} and in some cases even EFX~\cite{AmanatidisBFHV20}. An important question is, thus, if we can efficiently compute or approximate an allocation that maximizes \NSW. This is the question we study in this paper.

More formally, we consider an allocation problem with a set $\calN$ of $n$ agents and a set $\cal G$ of $m$ indivisible goods. Each agent $i \in \calN$ has a valuation function $v_i : 2^{\calG} \to \QQ_{\ge 0}$. We assume all functions to be non-negative, non-decreasing, and normalized to $v_i(\emptyset) = 0$. For convenience, we assume every $v_i$ maps into the rational numbers, since for computation these functions are part of the input. The goal is to find an allocation of the goods $\calA = (A_1,\ldots,A_n)$ to maximize the Nash social welfare, given by the geometric mean of the valuations
\[
	\NSW(\calA) = \left(\prod_{i=1}^n v_i(A_i) \right)^{1/n} \enspace.
\]
Clearly, an allocation that maximizes the \NSW\ is Pareto-optimal. By maximizing the \NSW, we strike a balance between maximizing the sum-social welfare $\sum_i v_i(A_i)$ and the egalitarian social welfare $\min_i v_i(A_i)$. Notably, optimality and approximation ratio for \NSW\ are invariant to scaling each valuation $v_i(A_i)$ by an agent-specific parameter $c_i > 0$. This is yet another property that makes \NSW\ an attractive objective function for allocation problems. It allows a further normalization -- we can assume every $v_i :  2^{\calG} \to \NN_0$ maps into the \emph{natural} numbers. 

Maybe surprisingly, finding desirable approximation algorithms for maximizing the \NSW\ has recently become an active field of research. For instances with additive valuations, where $v_i(A) = \sum_{j \in A} v_{ij}$ for every $i \in \calN$, in a series of papers~\cite{ColeDGJMVY17,ColeG18,AnariGSS17,BarmanKV18} several algorithms with small constant approximation factors were obtained. The currently best factor is $e^{1/e} \approx 1.445$~\cite{BarmanKV18}. The algorithm uses prices and techniques inspired by competitive equilibria, along with suitable rounding of valuations to guarantee polynomial running time.

Even for identical additive valuations, the problem is \classNP-hard, and a greedy algorithm with factor of 1.061~\cite{BarmanKV18AAMAS} as well as a PTAS~\cite{NguyenR14} were obtained. In terms of inapproximability, the best known lower bound for additive valuations is $\sqrt{8/7}\approx 1.069$~\cite{GargHM18}. Notably, this lower bound applies even in the case when the additive valuation is composed of only three values with one of them being 0 (i.e., $v_{ij} \in \{0,p,q\}$ for all $i \in \calN$, $j \in \cal G$, where $p,q \in \NN$). For the case of two values with one $0$ and one positive value (i.e., $v_{ij} \in \{0,q\}$  for all $i \in \calN$, $j \in \cal G$), an allocation maximizing the NSW can be computed in polynomial time~\cite{BarmanKV18AAMAS}. 

\paragraph{Contribution and Results.}
In this paper, we consider computing allocations with (near-)optimal \NSW\ when every agent has a \emph{2-value} valuation. In such an instance, $v_{ij} \in \{p,q\}$ for every $i \in \calN$ and $j \in \calG$, where $p,q \in \NN_{0}$. Notably, in 2-value instances any optimal allocation satisfies EFX, which is not true when agents have 3 or more values~\cite{AmanatidisBFHV20}. The case $p = q$ is trivial. An optimal allocation can be computed in polynomial time when $p = 0 < q$~\cite{BarmanKV18AAMAS}. Hence, we concentrate on the case $1 \le p < q$. Maybe surprisingly, we design a polynomial-time algorithm to find an optimal allocation when $p$ divides $q$, i.e., after appropriate scaling, when $p = 1$ and $q \in \NN$. Even if $p$ does not divide $q$, the algorithm still guarantees an approximation factor of at most $1.035$. This is drastically lower than the constant factors obtained for general additive valuations~\cite{ColeDGJMVY17,ColeG18,BarmanKV18}. An approximation algorithm for 2-value instances with approximation factor $1.061$ has been obtained in~\cite{GargM21}. The algorithm is based on ideas from competitive equilibria. Our algorithm is a greedy procedure and improves this guarantee.

Complementing these positive results, we also prove new hardness results for 2-value instances. Maximizing the \NSW\ is \classNP-hard whenever $p$ and $q$ are coprime and $p \ge 3$. Since for $p=1$ we have a polynomial-time algorithm, $p=2$ remains as an interesting open problem. Maximizing the \NSW\ in 2-value instances can even be \classAPX-hard. Our reduction from Gap-4D-Matching avoids the use of utilities $v_{ij} = 0$, which poses a substantial technical challenge over the more direct reduction for 3-value instances in~\cite{GargHM18}. Our lower bound on the approximation factor is $1.000015$ for $p/q = 4/5$. This answers an open problem from~\cite{AmanatidisBFHV20}.

\subsection{Related Work}

In addition to additive valuations, the design of approximation algorithms for maximizing \NSW\ with \emph{submodular valuations} has been subject to significant progress very recently. 
While small constant approximation factors have been obtained for special cases~\cite{GargHM18,AnariMGV18} (such as a factor $e^{1/e}$ for capped additive-separable concave~\cite{ChaudhuryCGGHM18} valuations), 
(rather high) constants for the approximation of \NSW\ with Rado valuations~\cite{GargHV21} and also general non-negative, non-decreasing submodular valuations~\cite{LiV21} have been obtained.

Interestingly, for \emph{dichotomous} submodular valuations where the marginal valuation of every agent for every good $j$ has only one possible non-negative value (i.e., $v_i(S \cup \{j\}) - v(S) \in \{0,p\}$ for $p \in \NN$), an allocation maximizing the \NSW\ can be computed in polynomial time~\cite{BabaioffEF21}. 
In particular,
in this case one can find in polynomial time an allocation that is Lorenz dominating, and simultaneously minimizes the lexicographic vector of valuations,
and maximizes both sum social welfare and Nash social welfare.
Moreover, this allocation also has favorable incentive properties in terms of misreporting of agents.

More generally, approximation algorithms for maximizing \NSW\ with subadditive valuations~\cite{BarmanBKS20,ChaudhuryGM21} and asymmetric agents~\cite{GargKK20} have been obtained, albeit thus far not with constant approximation ratios.

\section{Preliminaries}

An instance $\instance$ is given by the triple $(\agents, \goods, \set{v_i}_{i\in \agents})$ where $\agents$ is a set of $n$ agents and  $\goods$ is a set of $m\geq n$ indivisible goods. Every agent $i\in \agents$ has an additive valuation function with $v_i(A) = \sum_{j \in A} v_{ij}$ for every $A \subseteq \goods$. Here $v_{ij}$ represents the value $i$ assigns to the good $j \in \goods$. We assume that all $v_{ij} \ge 0$. In this paper, we study \emph{2-value} additive valuations, in which $v_{ij} \in \set{p,q}$ for $p,q \in \NN_0$. To avoid trivialities, we assume $0 < p < q$. Note that for $p = 0$ we recover the dichotomous case studied in~\cite{BarmanKV18AAMAS,BabaioffEF21}. We scale down the valuation of every agent by $q$ such that $v_{ij} \in \set{v,1}$ where $0 < v = p/q < 1$. Moreover, throughout the paper we assume $p$ and $q$ are coprime.

An allocation $A=(A_1, \dots, A_n)$ is a partition of $\goods$ among the agents, where $A_i \cap A_j = \emptyset$, for each $i \neq j$, and $\bigcup_{i\in \agents} A_i = \goods$. We evaluate an allocation using the Nash social welfare $\NSW(A) = \left(\prod _{i\in N}v_i(A_i)\right)^{\frac{1}{n}}.$

We represent every valuation by distinguishing between \emph{big} and \emph{small} goods for that agent. We use sets $B_i = \{ j \mid v_{ij} = 1\}$ and $S_i = \goods \setminus B_i$ to denote the subsets of goods that agent $i$ considers as big and small, respectively. Globally, we use $B = \bigcup_i B_i$ and $S = \bigcap_i S_i = \goods \setminus B$ for the sets of goods that are big for at least one agent or small for all agents, respectively. As such, an instance $\instance$ with 2-value additive valuations can be fully described by the vector $(\agents, \goods, (B_i)_{i \in \agents}, v)$. 

Of particular interest will be \emph{non-wasteful} allocations (c.f.~\cite{BabaioffEF21}), in which we only assign the goods from $B$ and give them to agents that value them as big goods. Formally, a non-wasteful allocation $A^b =(A^b_1, \dots, A^b_n)$ has $A^b_i \subseteq B_i$ and $\bigcup A_i^b = B$.

\paragraph{Comparing Optimal Allocations.} 
In our analysis, we often compare optimal allocations for 2-value valuations to optimal allocations for the same 2-value valuations with $v$ replaced by 0. Given a fair allocation instance $\instance$, we denote by $O^*$ an optimal allocation and $\NSW(O^*)$ the NSW of an optimal allocation. Similarly, for any $\instance$ we consider a corresponding \emph{dichotomous} instance $\instance^{(d)} = (\agents, B, \{v^{(d)}_i\}_{i \in \agents})$ obtained by setting $v^{(d)}_{ij} = 0$ for all $i \in \agents$, $j \in S_i$. We use $O$ to denote an optimal allocation in the dichotomous instance $\instance^{(d)}$. In particular, if we cannot assign a big good to every agent in $\instance^{(d)}$, we assume $O$ assigns a big good to as many agents as possible, and it maximizes the NSW among this set of agents. Note that we assume the goods in $S$ that are small for all agents are never assigned in $O$, and as such we exclude them from consideration in $\instance^{(d)}$. Clearly, $O$ will be a non-wasteful allocation.

We denote by $b_i = \modulus{B_i \cap O_i}$ and by $b^*_i = \modulus{B_i \cap O^*_i}$ the number of big goods agent $i$ is receiving in $O$ and $O^*$, respectively. Also, $b$ and $b^*$ are used to represent the vectors of $b_i$ and $b_i^*$, respectively.

In $O^*$, by Pareto-optimality, each good must be assigned to an agent. However, for any 
agent $i$, the set of big goods in $O_i$ might not be a superset of the set of big goods in $O_i^*$.

\begin{example} \label{ex:OPTvsOPTstar} \rm
	Let $\instance$ be a fair allocation instance with $n=2$ agents, $m=5$ goods, and $v=2/3$. All agents have identical valuations. There are two big goods and three small goods. Then only two optimal allocations exist (with $\NSW=2$) obtained by assigning all big goods to one agent and all small goods to the other. However, for $v=0$, every optimal allocation assigns each agent one big good. 
\end{example}

In general, there is no simple direct connection between $O$ and $O^*$, not even between vectors $b$ and $b^*$. In order to simplify our proofs, we will assume that $\agents$ is numbered in non-increasing order of $b_i$'s and subject to that in non-increasing order of $b^*_i$'s, i.e., for $i, j \in \agents$, if $b_i < b_j$, or $b_i = b_j$ and $b^*_i < b^*_j$, then $j < i$. There can be many optimal solutions $O^*$. For a rigorous reasoning we pick $O^*$ based on a hierarchy of three criteria based on $O$: (1) $O^*$ maximizes the NSW (i.e., it is \emph{optimal}); among all these solutions it (2) maximizes the overlap in big goods $\sum_{i \in \agents} |O_i \cap O_i^*|$ (i.e., it is \emph{sum-closest} to $O$); among all these solutions it (3) maximizes lexicographically $(\modulus{O_i \cap O^*_i})_{i \in \agents}$ (i.e., it is \emph{sum-lex-closest} to $O$). Condition (3) is tied to the ordering of the agents, for which the tie-breaking in turn depends on $O^*$. Tie-breaking and lexicographic maximization allow a consistent choice of $O^*$, since both aim to maximize the number of big goods in $O^*$ for agents with small index.

Given this choice of $O^*$, we capture the relation to $O$ in a more structured fashion using the notion of a \emph{transformation graph}.

Let $A$ and $A'$ be two possible allocations. We denote by $\transGraph{A}{A'}$ the {\em transformation graph} from allocation $A$ to allocation $A'$. More formally, $\transGraph{A}{A'} = (\agents, E_{A \to A'})$ is a directed multigraph, where $\agents$ is the set of the vertices. Each edge $e = (i,j) \in E_{A \to A'}$ corresponds to some good $k \in A_i \cap A'_j$ and vice versa. We use the notation $g(e) = k$. Observe that $\transGraph{A'}{A}$ can be obtained by simply reversing all the directed edges in $\transGraph{A}{A'}$.

A path in $\transGraph{A}{A'}$ can be seen as a sequence of goods $(g(e_1), g(e_2), \dots , g(e_{k-1}))$ such that $e_j = (i_j, i_{j+1})$ and $g(e_j) \in A_{i_j} \cap A'_{i_{j+1}}$ for all $j=1,\dots,k-1$. We say we {\em trade (goods along) a path} if we remove $g(e_j) $ from $A_{i_j} $ and add it in $A_{i_{j+1}}$, for each $j=1, \dots, k-1$. Moreover, we say that a path is a {\em balancing path} if after trade the utilities of the interior agents remain unchanged, i.e., $v_{i_j}(g(e_j)) = v_{i_j}(g(e_{j-1}))$, for each $j = 2, \dots, k-1$. Observe that every edge in the transformation graph is a balancing path; moreover, every path contained in a balancing path is a balancing path as well.

In general, there exist four types of balancing paths.
A \emph{small-to-big} or \emph{\SB-balancing path} is a balancing path  $(g(e_1), g(e_2), \dots , g(e_{k-1}))$, where $g(e_1)\in S_{i_1}$ and $g(e_{k-1}) \in B_{i_k}$. {\BS}/{\SS}/{\BB}-balancing paths are defined accordingly. Finally, we will briefly pay attention to \BB-balancing paths starting and ending at the same agent -- we term them \emph{balancing cycles} (and omit the prefix \BB, since clear from context).

\paragraph{Preliminaries on $O$, $O^*$, and $\transGraph{O^*}{O}$.}
Given the pair of allocations $O$ and $O^*$ with vectors $b$ and $b^*$ for the numbers of big goods, the next lemmas reveal some interesting structure of $\transGraph{O^*}{O}$. Notice that by the properties of $O$ and $O^*$ described above, the graph $\transGraph{O^*}{O}$ neither has \SS- nor \BS-balancing paths. Moreover, it has no balancing cycles, since $O^*$ is optimal and sum-closest to $O$. We are particularly interested in all agents, for which the number of big goods assigned in $O$ and $O^*$ differ. These agents are inherently connected to each other in the transformation graph.

\begin{lemma}\label{lemma:reachableAgentMoreBigGoods}
For every agent $i$ with $b_i^* > b_i$ there is an agent $j$ with $b_j^*< b_j$ such that in $\transGraph{O^*}{O}$ there is a \BB-balancing path from $i$ to $j$.
\end{lemma}
\begin{proof}
	Assume for contradiction that it is not true. Then for every agent $j$ that is reachable from $i$ by a {\sf BB}-balancing path it holds that $b_j \leq  b^*_j$. We denote by $C$ the set of these reachable agents, $i$ included. Let $\delta_j^{-}$ and $\delta_j^{+}$ be the in/out-degree of $j$, respectively, in the subgraph of $\transGraph{O^*}{O}$ induced by $C$. Observe that $\delta_j^{+}$ is at least the number of big goods given away by $j$ when comparing $O^*$ to $O$. Similarly, $\delta_j^{-}$ is at most the number of big goods agent $j$ receives in this comparison. Thus, $b^*_j \leq b_j + \delta_j^{+} -\delta_j^{-}$ for each $j\in C$. By assumption $b_j \leq b^*_j$ for each $j\in C\setminus \set{i}$, we see $\delta_j^{+} -\delta_j^{-} \geq 0$. Since $\sum_{j\in C} \delta_j^{+} = \sum_{j\in C} \delta_j^{-} $, we have $b^*_i - b_i \leq \delta_i^{+} - \delta_i^{-} \leq 0$ -- a contradiction to $b_i  < b^*_i$. Hence, a reachable agent $j \in C$ with $b_j^* < b_j$ must exist.
\end{proof}

\begin{lemma}\label{lemma:reachableAgentLessBigGoods}
	For every agent $j$ with $b_j^* < b_j$ there is an agent~$i$ 
	\begin{enumerate}
		\item \label{SBpath} such that in $\transGraph{O^*}{O}$ there is an \SB-balancing path from $i$ to $j$, or
		\item \label{BBpath} with $b_i^* > b_i$ such that in $\transGraph{O^*}{O}$ there is a \BB-balancing path from $i$ to $j$.
	\end{enumerate}
\end{lemma}
\begin{proof}
We denote by $C$ the set of agents that can reach $j$ through a balancing path, $j$ included. Let $\delta_i^{-}$ and $\delta_i^{+}$ be the in/out-degree of agent $i$ in the subgraph induced by $C$. We show that if in the induced subgraph there is no \SB-balancing path ending in $j$, then there must exist agent $i \in C$ such that $b_i^* > b_i$. Then, since $i \in C$ and $\transGraph{O^*}{O}$ neither has \SS- nor \BS-balancing paths, there must be a \BB-balancing path in $\transGraph{O^*}{O}$.

First, suppose there is at least one ``\SB-edge'' among agents in $C$, i.e., a pair $(\ell,\ell') \in E_{O^* \to O}$ with $\ell, \ell' \in C$, $g(\ell,\ell') \in S_{\ell}$, and $g(\ell,\ell') \in B_{\ell'}$. Clearly, by construction there is an \SB- or \BB-balancing path from $\ell'$ to $j$. If it is a \BB-balancing path, by extending it with edge $(\ell,\ell')$ we obtain an \SB-balancing path from $\ell$ to $j$. Hence, property (1) of the lemma holds.

Now suppose there is no ``\SB-edge'' among agents in $C$. Then there can be no \SB-balancing path in the subgraph induced by $C$. For each $i\in C$, $\delta_i^{+}$ is at most the number of big goods given away by $i$ when comparing $O^*$ to $O$. Now in $O$ no small good is assigned. Hence, by the construction of $C$, we see $\delta_i^{-}$ exactly corresponds to the number of big goods agent $i$ receives when changing $O^*$ to $O$, for each $i\in C$. Similar to the proof of \Cref{lemma:reachableAgentMoreBigGoods}, the connection between degrees and the vectors $b$ and $b^*$ implies property (2) of the lemma.
\end{proof}

\section{An Optimal Algorithm when $p$ Divides $q$}

Consider algorithm \TwoValueApprox. In phase 1, it computes $O$, the optimal allocation in the corresponding dichotomous instance $\instance^{(d)}$. This can be done in polynomial time~\cite{BarmanKV18AAMAS}. Note that after phase 1, there can be agents with empty bundles. Then we assume $O$ maximizes the number of agents receiving at least one good. Moreover, restricting attention to the set of agents with nonempty bundles, $O$ maximizes the NSW among them. It is easy to see that an allocation $O$ with this property is computed both by the algorithm for dichotomous additive instances in~\cite{BarmanKV18AAMAS} and its' generalization to dichotomous submodular ones in~\cite{BabaioffEF21}.

For phases 2 and 3, the algorithm calls procedure \Balance. In phase 2, if there exist unassigned goods (i.e., goods that are small for all agents), they get assigned sequentially to an agent with the currently smallest valuation. 

Finally, in phase 3, big goods received by the agents may be reallocated and turned into small ones. In particular, we greedily move a big good from the agent with the highest valuation to an agent with the smallest valuation if and only if this move increases the \NSW.  
\begin{algorithm}[t]
	\SetNoFillComment
	\KwIn{A fair allocation instance $\instance = (\agents ,\goods, (B_i)_{i\in \agents}, v)$}
	\KwOut{An allocation $A=(A_1, \dots, A_n)$}
	\tcc{Phase 1: Find optimal allocation for dichotomous instance} 
	Compute an optimal allocation $O = (O_1, \dots, O_n)$ for $\instance^{(d)}$ for goods in $B$\\
	\tcc{Phases 2 and 3} 
	$A \leftarrow$ \textsc{Balance}$(\instance,O)$ \\
	\KwRet{A}\\
	\caption{Algorithm \TwoValueApprox \label{Algorithm}}	
\end{algorithm}
\begin{algorithm}[t]
	\SetNoFillComment
	\KwIn{A fair allocation instance $\instance = (\agents, \goods, (B_i)_{i\in \agents}, v)$ and a non-wasteful allocation $A^b =(A^b_1, \dots, A^b_n)$ (i.e. with $A^b_i \subseteq B_i$ and $\bigcup A_i^b = B$)}
	\KwOut{An allocation $A= (A_1, \dots, A_n)$ of all goods in $\goods$}
	\tcc{Phase 2: Adding only-small valued goods} 
	Let $A_i = A^b_i$ for all $i \in \agents$\;
	\While{there exists $g\in S$}{
	$i= \arg\min_j {v_j(A_j)}$\\
	$A_i \leftarrow A_i \cup \set{g}$ and $S \leftarrow S\setminus \set{g}$ 
	}
	\tcc{Phase 3: Local search} 
	$i_1 = \arg\max_j{v_j(A_j)}$ and $i_2 = \arg\min_j{v_j(A_j)}$ \\
	\While{moving a good $g\in A_{i_1}$ to $A_{i_2}$ strictly increases $\NSW(A)$}{
		$A_{i_1} \leftarrow  A_{i_1} \setminus \set{g}$ and $ A_{i_2} \leftarrow  A_{i_2} \cup \set{g}$\\
		$i_1 = \arg\max_j{v_j(A_j)}$ and $i_2 = \arg\min_j{v_j(A_j)}$
	}
	\KwRet{A}\\
	\caption{Algorithm \Balance \label{Balance}}
\end{algorithm}

\paragraph{Running Time.}
To bound the running time, we start by proving a lemma about properties of phases 2 and 3 of the algorithm. We denote by $i_1^t$ and $i_2^t$ the agents $i_1$ and $i_2$ in round $t$ of phase 3.

\begin{lemma}
	\label{lem:propertiesAlgo}
	The following properties hold during the execution of \Balance$(\instance,O)$:
	\begin{itemize}
	    \item Every agent $i$ with small goods has a valuation of at most $v_i(A_i) \le \min_{j \in \agents} v_j(A_j) + v$.
	    \item If a move in round $t$ of phase 3 strictly increases the NSW, then (1) $i_1^t$ only has big goods, (2) we never moved a good away from agent $i_2^t$ during earlier rounds $1,\ldots,t-1$ of phase 3, and (3) none of the goods $g \in A_{i_1^t}$ is big for $i_2^t$.
    \end{itemize}
\end{lemma}
\begin{proof}
	We show the properties inductively. For the base case, we observe that they hold before the beginning of phase 3 (i.e., in a round 0 of phase 3). Since we allocated small goods greedily in phase 2, all agents $i,j$ with small goods in their bundle satisfy $v_i(A_i) \le \min_{j \in \agents} v_j(A_j) + v$. The properties (1)-(3) hold trivially before the beginning of phase 3, since the condition is not fulfilled and there are no earlier rounds. In other words, Properties (1)-(3) are empty for $t = 0$.
	%
	
	Now assume that all properties hold until round $t-1 \ge 0$ of phase 3 and consider round $t$. Clearly, when the hypothesis holds, the valuation of agents that have only big goods is decreasing, while the valuation of agents with small goods is increasing. Moreover, since the target agent has minimum valuation and receives a small good, all agents with small goods have valuation at most $\min_{j \in \agents} v_j(A_j) + v$. As a consequence, if $i_1^t$ in round $t$ has a small good, then moving any good is not profitable for the NSW, so we have (1). Suppose we moved a good from $i_2^t$ in an earlier round $t'$, i.e., $i_2^t = i_1^{t'}$. By hypothesis for round $t'$, $i_1^{t'}$ was the agent with highest valuation and had only big goods in round $t'$; moreover, for rounds $t'+1,\ldots,t-1$, the valuation of $i_1^t$ and $i_2^t = i_1^{t'}$ at the beginning of round $t$ can differ by at most 1. Then a move of a good from $i_1^t$ to $i_2^t$ is not profitable, which proves (2). Finally, suppose $g \in A_{i_1^t}$ is big for $i_2^t$. Then, since $i_2^t$ never lost a good up to round $t$, the valuation of $i_2^t$ only increased since the end of phase 1. As such, it must have been possible to improve $O$, either in terms of the number of agents receiving big goods, or in terms of NSW of agents with non-empty bundles. This is a contradiction and shows (3).
\end{proof}

Algorithm \TwoValueApprox\ runs in polynomial time: Phase 1 runs in polynomial time~\cite{BarmanKV18AAMAS}, and Lemma~\ref{lem:propertiesAlgo} shows that \Balance$(\instance,O)$ (re-)allocates each good at most once.

\paragraph{Optimality.}
Let us now focus on the Nash social welfare of the final allocation. We show that the algorithm computes an optimal allocation when $p$ divides $q$, i.e., when $p=1$ and $q \in \NN$ (after scaling valuations). In this case, an integer number of small goods are exactly as valuable as a big one. This fact will be key to show the main result in this section.
\begin{theorem}\label{thm:algoOPTintegerCase}
	If $p = 1$ and $q \in \mathbb{N}$, then Algorithm~\TwoValueApprox\ computes an optimal allocation in polynomial time.
\end{theorem}
\Cref{prop:bestComplement} is the first step toward proving the theorem. It implies that \Balance$(\instance, O)$ maintains an optimal assignment for a fixed number of big goods assigned to each of the agents. Towards this end, consider a \emph{partial big-allocation} $A^P$ such that $A^P_i \subseteq B_i$ for all $i \in \agents$, i.e., in $A^P$ all agents only receive big goods. Since $A^P$ is partial, there might be unassigned goods $\goods^U = \goods \setminus \bigcup_i A^P_i$. Now consider a \emph{small-extension} $A$ of $A^P$ obtained by assigning each good $g \in \goods^U$ to some agent $i$ with $g \in S_i$. Note that if a good $g \in \goods^U$ is big for all agents, then $A^P$ does not have any small-extension. We use the notation $s_i = \modulus{A_i \cap S_i}$.

\begin{proposition}\label{prop:bestComplement}
    If $v_i(A_i) + v < v_j(A_j)$ implies $s_j=0$ for every $i,j \in \agents$, then $A$ is a small-extension of $A^P$ with maximum NSW.
\end{proposition}
\begin{proof} 
    Assume by contradiction that $A$ is not the best small-extension of $A^P$. Let $A^*$ be a small-extension of $A^P$ with largest $\NSW$ that is sum-closest to $A$ (i.e., maximizes $\sum_{i \in \agents} |A_i \cap A_i^*|$). We define $s^*_i = \modulus{A^*_i \cap S_i}$. If $A$ is not optimal, then there exists $i\in \agents$ such that $s_i < s^*_i$. As in the proof of \Cref{lemma:reachableAgentMoreBigGoods}, we see that there must be an \SS-balancing path in $G_{A^* \rightarrow A}$ from $i$ to $j$ with $s^*_j< s_j$. Observe that $s_j > 0$. Hence, there exists a way to trade along the path without changing the valuation of interior agents. Since $A^*$ is an optimal small-extension that is sum-lex-closest to $A$, this must be strictly profitable, so 
	\[
		v_i(A^*_i) \cdot v_j(A^*_j) > \left( v_i(A^*_i)- v\right) \cdot\left( v_j(A^*_j) +v\right) \enspace .
	\]
    Then, since $v>0$, this is equivalent to 
	$ v_j(A^*_j) + v > v_i(A_i^*)$.
	Since $A_i \cap B_i = A_i^* \cap B_i = A^P_i$ and $A_j \cap B_j = A_j^* \cap B_j = A^P_j$, we see that $v_j(A^*_j) \le v_j(A_j)-v$ and $v_i(A_i^*) \ge v_i(A_i)+v$. Putting it all together we get 
	\[
		v_j(A_j) \ge v_j^*(A_j^*) + v > v_i(A^*_i) \ge v_i(A_i) + v \enspace . 
	\]
	However, we have $s_j > 0$, a contradiction to the assumption of $A$ in the lemma.
\end{proof}

We observed in Lemma~\ref{lem:propertiesAlgo} that throughout \Balance$(\instance,O)$, all agents receiving small goods differ in valuation by at most $v$. This implies that when $v_i(A_i) + v < v_j(A_i)$ at any point during the algorithm, then $s_j=0$, i.e., $j$ has no small goods.

For the next proposition, we assume \Balance$(\instance, \cdot)$ is applied to a particular form of non-wasteful allocation, which will eventually result in an optimal allocation. Recall that numbers $b_i$ and $b_i^*$ refer to the number of goods that agent $i$ receives in allocations $O$ and $O^*$, respectively, and that agents are numbered in non-increasing order of their valuation in $O$ and then in $O^*$, i.e., if $i \le j$, then $(b_i > b_j)$ or $(b_i = b_j$ and $b_i^* \ge b_j^*)$.

\begin{definition}
    \label{def:well-structured}
    An allocation $\tilde{O}$ is said to be \emph{well-structured} if it is non-wasteful and there is some value $0 \le K \le n$ 
    s.t.\
    \begin{itemize}
    \item $\tilde{b}=(b_1,\dots, b_K, b^*_{K+1}, \dots, b^*_n)$,
    \item for each $i \leq K$ either $b_i > b_i^*$, or $b_i = b^*_i$ and there is $j \leq K$ with $b_j > b^*_j$ and $b^*_j < b^*_i$, 
    \item for each $i \le K$ and $j > K$, $b^*_i \geq b^*_j$. 
\end{itemize}
\end{definition}

\begin{proposition}\label{prop:sameNumberBigGoods} 
Let $\tilde{O}$ be any well-structured allocation. Then \Balance$(\instance,\tilde{O})$ computes an optimal allocation. 
\end{proposition}

\begin{proof}
    We denote by $m' = \sum_{i=1}^{K} b_i - \sum_{i=1}^{K} b^*_i$ the number of goods from $B$ assigned as small in $O^*$.

    We start with some structural observations. Suppose we remove an arbitrary set $G$ of $m'$ goods from $\bigcup_{i\leq K}\tilde{O}_i$ in such a way that the numbers of the remaining big goods for agents $i\leq K$ compose a permutation of $(b^*_1,\ldots, b^*_K)$. We then assign the goods in $S \cup G$ sequentially to an agent with the currently lowest valuation. Moreover, let us pretend for the moment that the goods in $S \cup G$ are small for all the agents, that means, we increase the valuation of the agents receiving them by $v$. By Proposition \ref{prop:bestComplement}, this will lead to an optimal small-extension and, since we start from a partial allocation inducing a permutation of $(b^*_1,\ldots, b^*_K)$, this must be an allocation with maximum \NSW. This has several implications:
	\begin{enumerate}
	\item The goods in this process are indeed small for any agent receiving it. Otherwise, the allocation could be Pareto-improved, contradicting the optimality of $O^*$. 
	
	\item All small goods in $O^*$ are allocated to agents $i$ with $i > K$. For contradiction, suppose agent $i \leq K$ receives a small good. As the small goods were allocated in turn to an agent with minimum valuation, we can assume that $i$ has $\min_{k \le K} b^*_k$ big goods. Thus $i$ must have given some big good away. Then exchanging this good with the small one Pareto-improves the allocation, contradicting the optimality of the allocation.

	\end{enumerate}

    We now show that \Balance\ indeed removes a set $G$ of $m'$ goods as described above.
    
    If $m' = 0$, the statement is trivial and the proposition follows. We denote by $\tilde{O}^t$ the allocation after the $t$-th round in \Balance\ (counting both phases 2 and 3) and $\tilde b^t$ the vector of big goods. We will show inductively that (1) in every round the number of big goods remain ``above'' $O^*$, i.e., there is a permutation $\sigma$ of $\{1,\ldots,K\}$ such that $b_{\sigma(i)}^* \le \tilde{b}_i^t$ for all $i \le K$; and (2) in phase 3 the agent with highest valuation is an agent $i \le K$. As the base case, consider $t=0$ before the start of phase 2. Clearly, (1) and (2) hold by assumption. 
    
    Suppose both properties hold until the end of some round $t < m' + |S| - 1$. Consider round $t+1$. By hypothesis there is a permutation $\sigma$ such that $b_{\sigma(i)}^* \le \tilde{b}_i^t$ for all $i \le K$ and $\sum_{i \in \agents} \tilde{b}_i^t > \sum_{i \in \agents} b_i^*$. This implies that $\tilde{O}^{t}$ is sum-closer to $O$ than $O^*$, hence cannot be optimal. Moreover, there is $i < K$ such that $\tilde{b}_i^t > b_{\sigma(i)}^*$. If for all $j \le K$ we remove $\tilde{b}_j^t - b^*_{\sigma(j)}$ goods and assign them  iteratively to the least-valuation agents $k > K$, the NSW becomes optimal and thus strictly improves. This implies that after round $t$ there is a move improving the \NSW, so \Balance\ will not terminate since it would execute another round of phase 3.
    
    Now consider $i$ as the highest-valuation agent at the end of round $t$. By (2) this is an agent $i \le K$. 
    
    Suppose round $t+1$ is in phase 2. Then $\sigma$ still fits, and (1) holds after round $t+1$. Suppose (2) does not hold, i.e., after round $t+1$ an agent $j > K$ has highest valuation. This agent must have received the small good in round $t+1$, so the valuations of all agents differ by at most $v$. Hence, phase 3 would not start if phase 2 ended after round $t+1$. However, since there is at least one agent $k \le K$ with $b_{\sigma(k)}^* < \tilde{b}^t_k$, we proved above phase 3 would start after round $t+1$, a contradiction.
    
    Now suppose round $t+1$ is in phase 3. If $b^*_{\sigma(i)} < \tilde b^t_i$, then $\sigma$ still fits, so let us assume that $b^*_{\sigma(i)} = \tilde b^t_i$. If there is $j \leq K$ such that $\tilde b^t_j = \tilde b^t_i$ and $b^*_{\sigma(j)} < \tilde b^t_j$, then $(i, j) \circ \sigma$ works. Let us assume that all agents with maximum valuation in $\tilde O^t$ have as many goods as in $O^*$. We have $\tilde b^{t+1}_i < b^*_{\sigma(i)}$ and $\sum_{j \leq K} b^*_j \leq \sum_{j \leq K} \tilde b^{t+1}_j$ (because $t+1 \leq m'$), so there is $j \leq K$ such that $b^*_{\sigma(j)} < \tilde b^{t+1}_j$. Since $j$ can not have maximum valuation in $\tilde O^t$, so $\tilde b^{t+1}_j \leq \tilde b^{t+1}_i$ Consider the allocation where every $k \leq K$ gives away $\max(0, \tilde b^{t+1}_k - b^*_{\sigma(k)})$ goods, except $j$ that gives $\tilde b^{t+1}_j - b^*_{\sigma(j)} - 1$ goods. This allocation differs in valuation profile from $O^*$ only by agents $i$ and $j$ (up to a permutation) and we have $b^*_{\sigma(j)} < \tilde b^{t+1}_j \leq \tilde b^{t+1}_i \leq b^*_{\sigma(i)}$, so this new allocation has higher \NSW\ than $O^*$, a contradiction to the optimality of $O^*$. This proves that (1) holds after round $t+1$.

    Suppose (2) does not hold, i.e., there is an agent $j > K$ with highest valuation. This agent must have a small good, since $b_i^* \ge b_j^*$ for all $i \le K$, $j > K$. Hence, at the end of round $t+1$, the valuations of all agents differ by at most $v$, and there is no improving move left for round $t+2$. If $t+1 < m' + |S|$ we have an agent $k \le K$ with $b_{\sigma(k)}^* < \tilde{b}^t_k$, and \Balance\ will execute another round in phase 3, a contradiction. 

    Note that the good moved in round $t+1$ must be given to an agent $j > K$ -- even if we expanded the set of goods removed from agents $1,\ldots,K$ from the ones in rounds $1,\ldots,t+1$ to a set $G$ of goods considered above, all goods would be given only to agents $j > K$.
    
    Finally, we consider the case $t = m' + |S| - 1$. Then after round $t+1$, we obtain a permutation $\sigma$ of $\{1,\ldots,K\}$ such that $b^*_{\sigma(i)} \leq \tilde b^{m'}_i$ for all $i \le K$. We also have $\sum_{i \leq K} \tilde b^{m'+\modulus S}_i = \sum_{i \leq K} b_i - m' = \sum_{i \leq K} b^*_i$. Hence, $\tilde b^{m'+\modulus S}_i = b^*_i$ for all $i \le K$. Thus, the set of removed goods is a set $G$ considered above, and as such the resulting allocation $\tilde O^{m'-\modulus S}$ is optimal. As a consequence, \Balance\ stops after this iteration and returns an optimal allocation.
\end{proof}

The proposition shows that if the allocation computed in phase 1 has suitable properties, then the allocation computed by \Balance\ is an optimal one. We now further compare $O$ and $O^*$ to better understand why the hypothesis of \Cref{prop:sameNumberBigGoods} is not always satisfied by $O$ and which conditions on $v = p/q$ are sufficient for it.

In $O$ the big goods are as evenly balanced as possible. When $v\neq 0$, an optimal allocation $O^*$ might require to make the big goods more unbalanced. In the next proposition, we examine the details of this observation. In case $1/v \in \mathbb{N}$, we observe that \Cref{prop:sameNumberBigGoods} holds, and thus Algorithm \ref{Algorithm} computes an optimal allocation. Recall that we assume agents to be numbered in non-increasing order of $b_i$. The following proposition holds even when $O^*$ is optimal and sum-closest to $O$ (but not necessarily sum-lex-closest).

\begin{proposition}\label{prop:bigGoodsInOpt}
	Suppose $O^*$ is optimal and sum-closest to $O$ and there is an agent $i$ such that $b_i < b_i^*$. Consider an agent $j$ such that $b_j^* < b_j$ and there is a \BB-balancing path in $\transGraph{O^*}{O}$ from $i$ to $j$. 
	Then 
    \[
    v_i(O^*_i) -1 + v \cdot \lfloor 1/v \rfloor  \; < \;  v_j(O^*_j) \; < \; v_i(O^*_i) - 1 +  v \cdot \lceil 1/v \rceil \enspace,    
	\]
	as well as $b_j \leq b_i + 1$ and $b_j \le b_i^*$.
\end{proposition}
\begin{proof}

    For $k \in \agents$, we denote by $s^*_k = \modulus{O^*_k \cap S_k}$ the number of goods of $k$ that are small to $k$.

	As $O^*$ is optimal, trading along a \BB-balancing path in $\transGraph{O^*}{O}$ from $i$ to $j$ cannot increase the \NSW, i.e. $v_i(O^*_i)\cdot v_j(O^*_j) \ge (v_i(O^*_i) - 1) \cdot (v_j(O^*_j) + 1)$ and, hence, $v_i(O^*_i) \leq v_j(O^*_j) +1$, leading to the optimality condition $b^*_i + v s^*_i \leq b^*_j + v s^*_j + 1$. Besides, if $j$ has a good that is big to $i$, then either there is a balancing cycle, which contradicts the fact that $O^*$ is closest to $O$, or the good is small for $j$ and trading along the cycle gives a new allocation that Pareto-dominates $O^*$. So none of the goods of $j$ is considered big by $i$.

	We first show that $b_j \leq b_i+1$. Suppose for contradiction that this is not the case. Then by reversing the path between $i$ and $j$ and trading goods, we see that $O$ is not optimal in the dichotomous instance.
	
	Next we show $b^*_i - b^*_j \geq 2$ and $b_j \le b_i^*$. If $b_j \le b_i$, then $b^*_j < b_j \le b_i < b^*_i$ and since these numbers are integers we obtain $b^*_i - b^*_j \geq 2$, as well as $b_j \le b_i^*$. Thus, we are left with the case $b_j = b_i+1$. We have $b^*_j \leq b_i$ and $b_j \leq b^*_i$, and thus the following inequalities: $b^*_j \leq b_i = b_j-1 \leq b^*_i-1$. If one of the inequalities is strict, then we obtain $b^*_j \leq b^*_i-2$ and $b_j \le b_i^*$. Otherwise, $b^*_j = b_i$ and $b_j = b^*_i$. Then the optimality condition gives $s^*_i \leq s^*_j$. Now we trade along the path. Thereby we assign a big good to $j$. In exchange, agent $i$ receives $s^*_j - s^*_i $ many small goods from $j$'s bundle. This exchanges $v_i(O^*_i)$ and $v_j(O^*_j)$, and thus does not impact the \NSW. This contradicts the fact that $O^*$ is closest to $O$.
	
	Having shown that $b^*_i - b^*_j \geq 2$, we see with the optimality condition that $s^*_j -s^*_i \geq \frac{1}{v}$. We prove by contradiction that the relation between $v_j(O^*_j)$ and $v_i(O^*_i)$ holds. 
	
	Assume $v_j(O^*_j) \leq v_i(O^*_i)-1+ v\cdot \lfloor 1/v \rfloor$. Then
	\begin{dmath*}
	    v_i(O^*_i) \cdot v_j(O^*_j) \leq (v_i(O^*_i)-1+ v\cdot \lfloor 1/v \rfloor) \cdot (v_j(O^*_j)+1- v \cdot \lfloor 1/v \rfloor),
	\end{dmath*}
	which means that trading along the path from $i$ to $j$ and transferring $\lfloor 1/v \rfloor$ small goods from $O^*_j$ to $O_i$ does not decrease the $\NSW$ of the allocation. This is impossible because $O^*$ was taken as close to $O$ as possible.
	
	Now, if $v_j(O^*_j) \geq v_i(O^*_i)-1+ v \cdot \lceil 1/v \rceil$, then 
	\begin{dmath*} 
	    v_i(O^*_i) \cdot v_j(O^*_j) \leq (v_i(O^*_i)-1+v\cdot \lceil 1/v \rceil) \cdot (v_j(O^*_j)+1-v \cdot \lceil 1/v \rceil)
	\end{dmath*}
	and same reasoning applies by using $\lceil 1/v \rceil$ small goods.
\end{proof}

\Cref{lemma:reachableAgentMoreBigGoods} can be combined with \Cref{prop:bigGoodsInOpt} to yield the following corollary.

\begin{corollary}
	\label{cor:bigGoodsInOpt}
	If there is an agent $i$ with $b_i < b_i^*$, then there is some agent $j$ reachable from $i$ by a \BB-balancing path in $\transGraph{O^*}{O}$. Every such agent $j$ has at least one small good in $O^*_j$.
\end{corollary}

We can now prove \Cref{thm:algoOPTintegerCase}.
\begin{proof}[Proof of \Cref{thm:algoOPTintegerCase}]
We show that \Balance$(\instance,O)$ is an optimal allocation.
To this aim we show that $O$ satisfies the assumptions of \Cref{prop:sameNumberBigGoods}.

We first observe that if $1/v \in \mathbb{N}$, then there exists no agent $i$ such that $b_i< b_i^*$. Otherwise, by \Cref{lemma:reachableAgentMoreBigGoods} and \Cref{prop:bigGoodsInOpt}, there must exist an agent $j$ such $v_i(O^*_i) -1 + v \cdot \lfloor 1/v \rfloor < v_j(O^*_j) < v_i(O^*_i) - 1 + v \cdot \lceil 1/v \rceil$. Since, $1/v\in \mathbb{N}$, we have $\lfloor 1/v \rfloor = \lceil 1/v \rceil =1/v$ implying $ v_i(O^*_i) < v_j(O^*_j) +1 - v \cdot \lfloor 1/v \rfloor < v_i(O^*_i)$ which is impossible.
Thus, for each $i \in \agents$, $b_i\geq b_i^*$. Moreover, the entries of $b$ are sorted in non-increasing order.
By selecting $K$ as the maximum index $i \in \{0, \dots, n\}$ for which $b_i>b_i^*$, we see that $O$ is well-structured. Therefore, by \Cref{prop:sameNumberBigGoods}, \Balance$(\instance,O)$ returns an optimal allocation. 
\end{proof}

\section{Approximation}

In this section we study the case $1/v \not\in \mathbb{N}$ and prove a small approximation ratio for our algorithm. The idea is to compare the behavior of \Balance$(\instance,O)$ to \Balance$(\instance,\tilde{O})$ for a suitably chosen allocation $\tilde{O}$ such that the final allocation of the latter procedure is optimal. Towards the choice of $\tilde{O}$, we observe some additional properties of $O$ and $O^*$.

\begin{lemma}\label{lemma:eitherSBorBBpaths}
	For every agent $j\in \agents$ with $b_j^* < b_j$ exactly one of the following two properties hold:
	\begin{enumerate}
		\item \label{SBpath*} Either there is an agent $i$ and an \SB-balancing path from $i$ to $j$ in $\transGraph{O^*}{O}$,
		\item \label{BBpath*} or there is an agent $i$ such that $b_i^* > b_i$ and there is a \BB-balancing path from $i$ to $j$ in $\transGraph{O^*}{O}$.
	\end{enumerate}
\end{lemma}

\begin{proof}
	\Cref{lemma:reachableAgentLessBigGoods} implies that at least one of the two conditions hold. We now show that they cannot hold simultaneously.
	
	Let us assume for contradiction that there are $i,k\in \agents$ such that both conditions hold: there is an \SB-balancing path from $k$ to $j$, $b_i^* > b_i$, and there is a \BB-balancing path from $i$ to $j$. Since the assumptions of \Cref{prop:bigGoodsInOpt} are fulfilled, by \Cref{cor:bigGoodsInOpt}, $j$ has a small good in her bundle. Suppose we trade along the \SB-balancing path from $k$ to $j$, and $j$ gives one of her small goods to $k$ in turn. Then we obtain a new allocation that Pareto-dominates $O^*$ -- a contradiction.
\end{proof}

We now show how to transform $O^*$ into an allocation $\tilde{O}$ with corresponding vector of big goods $\tilde{b}$ such that for each $i\in \agents$ either $\tilde{b}_i=b_i$ or $\tilde{b}_i=b^*_i$. Let us denote by $\mathcal{P}_{\mathsf{SB}}$ the union of all \SB-balancing paths in $\transGraph{O^*}{O}$. Note that since $O^*$ is closest to $O$, every \SB-balancing path consists of a single edge. In order to get allocation $\tilde{O}$, we perform the following steps. We initialize $\tilde{O} \leftarrow O^*$ and 
\begin{itemize}
	\item for every edge $(i,j)\in \mathcal{P}_{\mathsf{SB}}$
	\begin{itemize}
		\item $\tilde{O}_i \leftarrow  \tilde{O}_i \setminus \set{g(i,j)}$
		\item $\tilde{O}_j \leftarrow  \tilde{O}_j \cup \set{g(i,j)}$
	\end{itemize}
	\item remove all the goods $g \in S$.
\end{itemize}

The result of this procedure is an allocation $\tilde{O}$, in which each good from $B$ that is assigned as small in $O^*$ is moved to the agent that owns it in $O$. 
Finally, all goods from $S$ are removed.

Observe that no \BB-balancing cycles exist in $\transGraph{O^*}{O}$ since $O^*$ is closest to $O$. Since $\transGraph{\tilde O}{O}$ is a subgraph of $\transGraph{O^*}{O}$, no \BB-balancing cycle can emerge in $\transGraph{\tilde O}{O}$. Moreover, for each $i\in \agents$ we have $\tilde{b}_i \ge b^*_i$. In particular, $\tilde{b}_i=b_i$ or $\tilde{b}_i=b^*_i$, or both (in case $b_i = b_i^*$). Indeed, for each agent with $b_i > b_i^*$ that was reachable by an \SB-balancing path, we know that, by \Cref{lemma:eitherSBorBBpaths}, $\mathcal{P}_{\mathsf{SB}}$ contains $b_i - b_i^*$ \SB-balancing paths, one for each good in $O_i \setminus O_i^*$. Thus, in the first part of the procedure $i$ will receive  $b_i - b_i^*$ big goods.
All the other agents do not lose or gain any big good during the transformation of $O^*$ into $\tilde{O}$ and thus for such agents $\tilde{b}_i=b_i^*$.

The next lemma shows that the vector $\tilde{b}$ can be written as $(b_1,\dots, b_K, b^*_{K+1}, \dots, b^*_n)$ for some index $0 \le K \le n$.

\begin{lemma}\label{lemma:orderingTildeB}
	For every $i \in \agents$ with $\tilde{b}_i=b_i > b_i^*$ and every $j \in \agents$ with $\tilde{b}_j = b^*_j \neq b_j$, it holds $b_i > b_j$ and $b^*_i \geq b^*_j$.
\end{lemma}

\begin{proof} We split the proof into two cases.
	
	\paragraph{Case 1:}
	If $b_j < b^*_j$, then by~\Cref{lemma:reachableAgentMoreBigGoods},~\Cref{prop:bigGoodsInOpt} and~\Cref{cor:bigGoodsInOpt}, there is some agent $j' \in \agents$ such that $v_{j'}(O^*_{j'}) > v_j(O^*_j) - 1 + v \cdot \lfloor 1/ v \rfloor$ and $j'$ has a small good in $O^*_{j'}$. On the other hand, since $\tilde b_i = b_i > b_i^*$, by construction of $\tilde O$, there is an \SB-balancing path in $\transGraph{O^*}{O}$ from some agent $k$ to $i$. If $i$ has a small good in her bundle in $O^*$, we can trade along the path from $k$ to $i$ and then let agent $i$ give one small good in her bundle to $k$. In this way, we obtain an allocation that Pareto-dominates $O^*$. Hence, $i$ must have no small good in her bundle in $O^*$, so $v_i(O_i^*) = b^*_i$. Because of optimality of $O^*$, it is not strictly more profitable in terms of \NSW\ to move one small good of $j'$ to the bundle of $i$, implying $v_{j'}(O^*_{j'}) - v \leq v_i(O_i^*)$ and thus,
	\begin{dmath*}
	b^*_i + v = v_i(O^*_i) + v \ge v_{j'}(O^*_{j'}) > v_j(O^*_j) - 1 + v \cdot \lfloor 1/ v \rfloor \geq b^*_j - 1 + v \cdot \lfloor 1/ v \rfloor \enspace . 
	\end{dmath*}
	Thus, $b^*_j - b^*_i < 1+v- v\cdot  \lfloor 1/ v \rfloor \leq 1$. Since $b_i^*$ and $b_j^*$ are integers, we see that this implies $b_i^* \ge b_j^*$ and hence $b_i > b_j$.

	\paragraph{Case 2:} 
	We now assume $b_j > b_j^* $. Observe that we removed all \SB-balancing paths of $\transGraph{O^*}{O}$ in the transformation of $O^*$ into $\tilde{O}$ but we still have $\tilde b_j = b^*_j$.
	Hence, there exists some agent $j'$ with $b_{j'} < b_{j'}^*$ and  a \BB-balancing path in $\transGraph{O^*}{O}$ from $j'$ to $j$. 
	\Cref{prop:bigGoodsInOpt} gives $b^*_{j'} \geq b_j$.
	Finally, since $b_{j'} < b_{j'}^*$ it must be that $ \tilde{b}_{j'}= b_{j'}^*$. We can apply Case 1 on $j'$ and obtain $b_i >b^*_i \ge b^*_{j'} \ge b_j > b_j^*$ and the claim follows.
\end{proof}

 We set $K$ to the largest index such that $\forall i \leq K$ we have $\tilde b_i = b_i > b^*_i$, or $b_i = b^*_i$ and there is $j \leq K$ such that $\tilde b_j = b_j > b^*_j$ and $b^*_j < b^*_i$. If there is no such index, we simply set $K = 0$. Intuitively, we choose $K$ as the largest index such that $\tilde O$ qualifies as a well-structured allocation in the sense of \Cref{def:well-structured}. 
\begin{lemma}
    \label{lemma:tildeOoptimal}
    $\tilde{O}$ is well-structured.
\end{lemma}

\begin{proof}
 First, let $i \in \agents$ such that $\tilde b_i = b_i > b^*_i$. We show that $i$ satisfies the condition defining $K$. Let $j \leq i$. If $\tilde b_j = b_j > b^*_j$, then there is nothing to show. If $b_j = b^*_j$, then because of the ordering of $\agents$ we have $\tilde b_j = b^*_j = b_j \geq b_i > b^*_i$. The only case left is $\tilde b_j = b^*_j \neq b_j$. In this case,~$\Cref{lemma:orderingTildeB}$ tells us that $b_i > b_j$, a contradiction to the ordering on $\agents$. Thus $i \leq K$.

Finally, it remains to consider $i \leq K$ and $j > K$ and show that $b^*_j \leq b^*_i$. If $\tilde b_j = b_j > b^*_j$, then $j \le K$, which is not the case. If $\tilde b_j = b^*_j \neq b_j$, then by~$\Cref{lemma:orderingTildeB}$, we have $b^*_i \geq b^*_j$. The only case left is $b_j = b^*_j$.
If $b_i = b_i^*$ this is trivial by the ordering of agents, so we assume $b_i > b_i^*$. Now if $b^*_j > b^*_i$, then let $k \in \agents$ such that there exists an \SB-balancing path in $\transGraph{O^*}{O}$ from $k$ to $i$. Note that $i$ cannot have a small good in $O^*$ -- otherwise we could trade along the \SB-balancing path, give $i$'s small good to $k$ and obtain an allocation that Pareto-dominates $O^*$. Hence $v_i(O^*_i) = b_i^*$. But then trading along the \SB-balancing path and giving a big good of $j$ to $k$ does not decrease the NSW but increases $\modulus{O_i \cap O^*_i}$ by one and decreases $\modulus{O_j \cap O^*_j}$ by one, a contradiction to the fact that $O^*$ is closest to $O$ since $i < j$.
\end{proof}

Now, with \Cref{prop:sameNumberBigGoods}, we have that \Balance$(\instance, \tilde{O})$ returns an optimal allocation. Suppose we run \Balance$(\instance, \tilde{O})$. Let $\tilde{O}^t$ denote the allocation and $\tilde b^t$ the vector of big goods after $t$ rounds of phase 3, and let $\tilde{T}$ be the last step before \Balance$(\instance, \tilde{O})$ terminates. The previous lemma shows that the allocation $\tilde{O}^{\tilde{T}}$ is an optimal allocation (possibly different from $O^*$). \Balance\ moves big goods from agents $i \le K$ and assigns them as small to agents $j > K$ as long as it is strictly profitable for the \NSW. For this reason, for every agent $j > K$ the number of big goods stays the same during the procedure. 
In $\tilde{O}^{\tilde{T}}$, every agent $j > K$ has $b_j^*$ big goods, resp. for the agents $i \le K$, the numbers of big goods can be different from $b_i^*$. Recall that there are no \BB-balancing cycles in $\transGraph{\tilde O}{O}$, and the execution of \Balance$(\instance, \tilde{O})$ will not introduce any of them. Moreover, the execution of \Balance$(\instance, \tilde{O})$ will not introduce \BB-balancing paths adjacent to agents $i \le K$. Hence, all \BB-balancing paths in $\transGraph{\tilde{O}^{\tilde{T}}}{O}$
also exist in $\transGraph{O^*}{O}$. Recall that $O^*$ is an optimal allocation sum-lex-closest to $O$, which implies that $\tilde{O}^{\tilde{T}}$ is an optimal allocation that is sum-closest to $O$. Note that this allows us to apply \Cref{prop:bigGoodsInOpt} and \Cref{cor:bigGoodsInOTilde} with $\tilde{O}^{\tilde{T}}$ instead of $O^*$.

\begin{corollary}
    \label{cor:bigGoodsInOTilde}
	Suppose there is an agent $i$ such that $b_i < \tilde{b}_i^{\tilde{T}}$. There is at least one agent $j$ such that $\tilde{b}_j^{\tilde{T}} < b_j$ and there is a \BB-balancing path in $\transGraph{\tilde{O}^{\tilde{T}}}{O}$ from $i$ to $j$. For every such agent $j$, we have 
    \[
	v_i(\tilde{O}^{\tilde{T}}_i) -1 + v \cdot \lfloor 1/v \rfloor  \; < \;  v_j(\tilde{O}^{\tilde{T}}_j) \; < \; v_i(\tilde{O}^{\tilde{T}}_i) - 1 +  v \cdot \lceil 1/v \rceil \enspace,
	\]
	as well as $b_j \leq b_i + 1$, $b_j \le \tilde{b}_i^{\tilde{T}}$, and there is at least one small good in $\tilde{O}^{\tilde{T}}_j$.
\end{corollary}

We derive another useful property when starting \Balance\ from $\tilde{O}$.

\begin{lemma}\label{lemma:MaxMinInTildeO}
At any time step $t \leq \tilde{T}$ of  Phase 3 in \Balance$(\instance, \tilde{O})$ and  for any agent $i$ who receives a small good, $\min\limits_{k \leq K} b^*_k >  v_i(\tilde{O}^t_i)$ holds.
\end{lemma}

\begin{proof}
    Let $I_1^t$ and $I_2^t$  be the set of agents giving goods away and of agents receiving small goods during rounds $1,\ldots,t$ in phase 3, respectively. 
    
    We first show $\min_{i \in I_1^t} v_i(\tilde{O}^t_i) > \max_{i \in I_2^t} v_i(\tilde{O}^t_i)$.
    
    Assume towards a contradiction that we have $i \in I_1^t$ and $j \in I_2^t$ such that $v_i(\tilde{O}^t_i) \leq v_j(\tilde{O}^t_j)$. Let $s_1$ and $s_2$ denote the latest iterations such that $i = i_1^{s_1+1}$ and $j = i_2^{s_2+1}$. We have $v_i(\tilde{O}^{s_1}_i) = v_i(\tilde{O}^t_i)+1$ and $v_j(\tilde{O}^{s_2}_j) = v_j(\tilde{O}^t_j)-v$. First observe a condition for the improvement of the NSW when a big good becomes a small good: for $m, M \in \mathbb R$, we have $(M-1)(m+v) > Mm$ if and only if $m/v < M-1$, since $v \in (0,1)$. 
    
    Now if $s_2 < s_1$, let $m = \min_{k \in \agents} v_k(\tilde{O}^{s_1}_k)$, $M = \max_{k \in \agents} v_k(\tilde{O}^{s_1}_k)$, and $a, b \in \NN$ such that $m = av+b$. Then $m+v \geq v_j(\tilde{O}^{s_1}_j) = v_j(\tilde{O}^t_j) \geq v_i(\tilde{O}^t_i) = v_i(\tilde{O}^{s_1}_i)-1 = M-1$. In round $s_1 + 1$ the NSW improves by moving a big good away from $i = i_1^{s_1+1}$. With the improvement condition \[a+b < m/v < M-1 \leq m+v = (a+1)v+b < a+b+1,\]
    but $a+b < M-1 < a+b+1$ is impossible since the numbers are integers. This yields a contradiction to $v_i(\tilde{O}^t_i) \leq v_j(\tilde{O}^t_j)$.
    
    If $s_1 < s_2$, let $m = \min_{k \in \agents} v_k(\tilde{O}^{s_2}_k)$, $M = \max_{k \in \agents} v_k(\tilde{O}^{s_2}_k)$, and $a, b \in \NN$ such that $m = av+b$. Then $M-1 \leq v_i(\tilde{O}^{s_2}_i) = v_i(\tilde{O}^t_i) \leq v_j(\tilde{O}^t_j) = v_j(\tilde{O}^{s_1}_j)+v = m+v$. Again, we obtain a contradiction with $a+b < M-1 < a+b+1$.
    
    Finally, if $s_1 = s_2$, let $m = \min_{k \in \agents} v_k(\tilde{O}^{s_2}_k)$, $M = \max_{k \in \agents} v_k(\tilde{O}^{s_2}_k)$, and $a, b \in \NN$ such that $m = av+b$. Then $m+v = v_j(\tilde{O}^{s_2}_j)+v = v_j(\tilde{O}^t_j) \ge v_i(\tilde{O}^t_i) = v_i(\tilde{O}^{s_1}_i) -1 = M-1$. Again, we obtain a contradiction with $a+b < M-1 < a+b+1$.
    
    In conclusion, $\min_{i \in I_1^t} v_i(\tilde{O}^t_i) > \max_{i \in I_2^t} v_i(\tilde{O}^t_i)$ for any $t\leq \tilde{T}$.
    
    Since the valuation of the agents in $ I_1^{\tilde{T}}$ are not increasing and the ones of the agents in $ I_2^{\tilde{T}}$ are not decreasing during the Phase 3  we then also have $\min_{i \in I_1^{\tilde{T}}} v_i(\tilde{O}^{\tilde{T}}_i) > \max_{i \in I_2^t} v_i(\tilde{O}^t_i)$. To conclude, we recall \Cref{prop:sameNumberBigGoods} shows that $(v_k(\tilde{O}^{\tilde{T}}_k))_{k=1,\ldots,K}$ is a permutation of $(b_k^*)_{k=1,\ldots,K}$.   As a consequence, $ \min_{i \in I_1^{\tilde{T}}} v_i(\tilde{O}^{\tilde{T}}_i)\geq \min\limits_{k \leq K} b^*_k$ as $I_1^{\tilde{T}}\subseteq \set{1, \dots, K}$. Moreover, if $ \min_{i \in I_1^{\tilde{T}}} v_i(\tilde{O}^{\tilde{T}}_i)> \min\limits_{k \leq K} b^*_k$, since every agent in $I_1^{\tilde{T}}$ gave away at least one good from their initial bundle, then $b_h -1> \min\limits_{k \leq K} b^*_k$ for each $h \in I_1^{\tilde{T}}$. Hence, by definition of $K$, $ \arg\min\limits_{k \leq K} b^*_k > K$,  a contradiction. In conclusion, $ \min_{i \in I_1^{\tilde{T}}} v_i(\tilde{O}^{\tilde{T}}_i)= \min\limits_{k \leq K} b^*_k$, thus, for any $i \in I^t_2$, we have $\min\limits_{k \leq K} b^*_k >  v_i(\tilde{O}^t_i)$.
\end{proof}

To show the approximation factor of our algorithm, we relate $\tilde{O}^{\tilde{T}}$ to the output of our algorithm, i.e., the output of \Balance$(\instance, \tilde{O})$ to the one of \Balance$(\instance, O)$. For this purpose, we track the allocations in \Balance$(\instance, \tilde{O})$ and simultaneously apply them on $O$. Let $O^t$ denote the allocation and $b^t$ the vector of big goods after $t$ rounds of phase 3. We couple the changes to big goods in $\tilde{O}^t$ and $O^t$ in the following way: 1) In a step $t$ of phase 2, a globally small good from $S$ is added to both $O^t$ and $\tilde{O}^t$. It is given to an agent with the current smallest valuation in the respective allocation. 2) In a step $t$ of phase 3, in which a big good is removed from the bundle of agent $i\leq K$ in $\tilde{O}^t$, we also remove one big good from $i$'s bundle in $O^t$. The good is given to an agent with the current smallest valuation in the respective allocation. 
Note that we couple the removal of the big good, but as small good it then gets assigned to potentially different agents in $O^t$ and $\tilde{O}^t$.

Let $T$ be the final step of \Balance$(\instance, O)$. Observe that in every step $t \leq T$, we can assume that the coupled process on $O$ behaves exactly like \Balance$(\instance, O)$. 
However, it might be that $\tilde{T} \neq T$. Then, if $\tilde{T} > T$, the coupled process forces \Balance$(\instance,O)$ to continue turning big goods into small ones although this is not profitable for the NSW.

We also observe that, if $i$ is the agent with current lowest valuation in $O^{\tilde{T}} $, then, she will receive a small good. More formally, we prove the following lemma. 

\begin{lemma}\label{lemma:senderNotReceiver}
    In the beginning of any round $t+1 \le \tilde{T}$, suppose $h$ is any agent with minimum valuation. Then $v_h(O^t_h) < \min\limits_{k \leq K} b^*_k$, and the good reallocated in round $t+1$ is small for $h$.
\end{lemma} 

\begin{proof}
We prove the statement by strong induction by assuming that it holds for all $t' < t$. We first show that $v_{h}(O^t_{h}) < \min\limits_{k \leq K} b^*_k$ for any agent $h$ with minimum valuation in $O^t$.

 Assume towards a contradiction that we have $\min\limits_{k \leq K} b^*_k \leq v_h(O^t_h) = \min\limits_{j \in \agents} v_j(O^t_j)$.

\Cref{prop:sameNumberBigGoods} shows that $(v_k(\tilde{O}^{\tilde{T}}_k))_{k=1,\ldots,K}$ is a permutation of $(b_k^*)_{k=1,\ldots,K}$. Thus, for all $i \le K$ and $t' < t$ we have $v_i(O^{t'}_i) \geq v_i(\tilde{O}^{t'}_i) \geq v_i(\tilde{O}^{\tilde T}_i) \geq \min\limits_{k \leq K} b^*_k$, i.e., so no such agent can have a minimum valuation before time step $t$, and thus none of them received a good. 
Thus, we have, for $i \leq K$, $v_i(O^t_i) = v_i(\tilde O^t_i)$, implying  $\sum\limits_{i \leq K} v_i(O^t_i) = \sum\limits_{i \leq K} v_i(\tilde O^t_i)$. Furthermore, until the end of round $t$ the (re-)allocated goods were always small for the agents receiving them, hence $\sum\limits_{j > K} v_j(O^t_j) = \sum\limits_{j > K} v_j(\tilde O^t_j)$.
	
Now look at $\tilde{O}^{t}$. For $j > K$ that did not receive any good, using the hypothesis we have
\[  
    v_j(\tilde O^{t}_j) = b^*_j \leq \min\limits_{k \leq K} b^*_k \leq v_h(O^{t}_h) = \min\limits_{k \in \agents} v_k(O^{t}_k) \enspace, 
\]
since $b^*_j \leq \min\limits_{k \leq K} b^*_k $ for $j >K$ by \Cref{lemma:orderingTildeB}. On the other hand, for any agent $j>K$ that received some small good up to round $t-1$, we
have by \Cref{lemma:MaxMinInTildeO} that 
$v_j(\tilde O^t_j)< \min\limits_{k \leq K} b^*_k \le v_h(\tilde O^t_h) = \min\limits_{k \in \agents} v_k(O^t_k)$. In conclusion,  $\sum\limits_{j > K} v_j(O^t_j) = \sum\limits_{j > K} v_j(\tilde O^t_j) <  \min\limits_{k \in \agents} v_k(O^t_k)\cdot(n-K+1)$ -- a contradiction.

Let us now show that the received good at time step $t+1$ is small for any minimum valuation agent $h$. Let $j$ be the highest valuation agent at time step $t+1$ who is giving away a good from her bundle. Since until time step $t+1$ no agent $i \leq K$ received a  good in the transformation of $O$ into $O^{\tilde{T}}$, the following inequalities hold 

	$
	    b_j \ge v_j(O^t_j) = v_j(\tilde O^t_j) > v_j(\tilde O^{\tilde T}_j) \ge \min\limits_{k \le K} b^*_k > v_h(O^t_h) \ge b_h \enspace.
	$

    Since $b_j$, $\min_{k \leq K} b^*_k$ and $b_h$ are integers, we have $b_j> b_h+1$.
	On the other hand, if a good in $O_j$ is big for $h$ then $b_h+1 \geq b_j$ as $O$ is an optimal allocation when $v=0$. In conclusion, we get $b_h+1 \geq b_j > b_h+1$ -- a contradiction.
\end{proof}

This lemma has several implications reported in the following.

\begin{corollary}\label{cor:behaviorOnO}
For any $t< \tilde{T}$, we have that at time step $t$
\begin{enumerate}
\item no agent $i\leq K$ becomes a minimum valuation agent in $O^t$,
\item no agent $j > K$ receives big goods,
\item no agent $i> K$ becomes a maximum valuation agent in $O^t$ .
\end{enumerate}
\end{corollary}

\begin{proof}
Statements (1) and (2) immediately follow by the proof of \Cref{lemma:senderNotReceiver}. 

We now show statement (3).
If $i$ has no small good, then $v_i(O^t_i) = b^*_i \leq \min\limits_{k \leq K} b^*_k$, and if $i$ has a small good, then $v_i(O^t_i) \le \min\limits_{j > K} v_j(O^t_j)+v < \min\limits_{k \le K} b^*_k + v$. Now, let $j < K$ be the agent that is giving away a good at time step $t+1$ with $t< \tilde{T}$, then since $j$ ends up with $v_j(\tilde O^{\tilde T}_j) \geq \min\limits_{k \le K} b^*_k $, we have  $v_j(O^t_j) = v_j(\tilde O^t_j) \ge \min\limits_{k \le K} b^*_k + 1 > v_i(O^t_i)$. So $i$ does not have maximum valuation.
\end{proof}

The next lemma shows that  $\NSW(O^{\tilde{T}}) $ is not higher than $ \NSW(O^T)$. 

\begin{lemma}\label{lemma:decreasingWelfare}
	 $\NSW(O^{\tilde{T}}) \leq \NSW(O^T)$.
\end{lemma}

\begin{proof}
	
	If $\tilde{T} \leq T$ the claim is immediate, since $O^T$ is the output of \Balance\ applied on $O$ and thus for any time step $t < T$, $ \NSW(O^t)$ is strictly increasing.
	
	Hence, we now consider the case $T < \tilde{T}$ and we inductively show $\NSW(O^{t}) \geq  \NSW(O^{t+1})$, for each $T \leq t < \tilde{T}$. 
	Let us denote by $v^{t}_{max}$ and $v^{t}_{min}$ the maximum and the minimum valuation in $O^t$.
	
	If $t=T$, since  \Balance$(\mathcal{\instance}, O)$ terminates at time step $T$  we have $\NSW(O^{T}) \geq \NSW(O^{T+1})$. 

    We now assume the statement true until $t$ and show it for $t+1 < \tilde{T}$, i.e., our inductive hypothesis is $\NSW(O^{t+1}) \leq \NSW(O^{t})$, which is equivalent to $v\cdot v^{t}_{max} \leq v^t_{min} + v$.
    
    \Cref{cor:behaviorOnO} implies that by coupling the processes, the bundles of agents $i \le K$ are the same for $O^t$ and $\tilde{O}^t$, even for $T \le t \le \tilde{T}$ and that we are moving a big good from the bundle of a maximum valuation agent to one of the agent with minimum valuation. Moreover, \Cref{{cor:behaviorOnO}} shows that the agent giving away the good does not become one with minimum valuation, so $v_{max}^t - 1 > v_{min}^{t+1}$. As the valuation of the agent that receives a good increases, we have $v_{min}^t \leq v_{min}^{t+1}$. Thus, $v_{min}^t < v_{max}^t - 1$, and the agent that receives a good increases his value by at most $v$, so $v_{max}^{t+1} \le v_{max}^t$. Overall, we have
	
	\[v \cdot v^{t+1}_{max} \leq v \cdot v^t_{max} \le v^t_{min} + v \leq v^{t+1}_{min} + v \enspace .\]

	This proves $\NSW(O^{t+2}) \leq \NSW(O^{t+1})$. 
\end{proof}

Finally, to bound the approximation factor of Algorithm \ref{Algorithm}, we show that we can partition the agents into two groups. In one group, the agents have the same valuation in $O^{\tilde{T}}$ and $\tilde{O}^{\tilde{T}}$. In the other, the following properties are satisfied:

\begin{itemize}
	\item the utilitarian social welfare is the same (in particular, the number of goods assigned as big/small is the same in the two allocations)
	\item in $O^{\tilde{T}}$ the valuations of any pair of agents differs by at most $v$ 
\end{itemize}

To this end, reconsider $O$ and $\tilde{O}$. We consider a permutation $\sigma$ of $\agents$ such that $\tilde{b}^\sigma = (\tilde{b}^\sigma_1, \dots, \tilde{b}^\sigma_n) = (\tilde{b}_{\sigma(1)}, \dots, \tilde{b}_{\sigma(n)})$ is non-increasing and that minimizes the exchanged indices. Note that the two vectors, $b$ and $\tilde{b}$, are component-wise the same for the first $K$ entries, and $\tilde{b}_i \geq \tilde{b}_j$, for $i\leq K$ and $j>K$. Hence, we have $\sigma(i)=i$ for $i\leq K$. 

\begin{proposition}\label{prop:utilitiesForApproximation}
	There exists $N'\subset \agents$ such that:
	\begin{enumerate}
		\item for each $i\in \agents \setminus N'$, $v_i(O_i^{\tilde{T}})= v_{\sigma(i)}(\tilde{O}_{\sigma(i)}^{\tilde{T}})$, 
		\item  $\sum_{i\in N'} v_i(O_i^{\tilde{T}}) = \sum_{i\in {N}'} v_{\sigma(i)}(\tilde{O}_{\sigma(i)}^{\tilde{T}})$, and 
		\item for each $i, j\in N'$ $\modulus{v_i({O}_i^{\tilde{T}})-v_j({O}_j^{\tilde{T}})} \leq v$.
	\end{enumerate}
\end{proposition}

\begin{proof}
We show there exists an index $h>K$ such that $N'= \set{h, h+1, \dots, n}$ fulfills the properties.
	
	Let $H$ be the minimum index for which $b$ and $\tilde{b}^{\sigma}$ differ. Observe that $H>K$.
	Let $h=\min \set{j_1, j_2, H}$, where $j_1$ (resp.\ $j_2$) is the smallest index for which $ O^{\tilde{T}}_{j_1}\cap S_{j_1} \neq \emptyset$ (resp.\ $\tilde{O}^{\tilde{T}}_{\sigma(j_2)} \cap S_{\sigma(j_2)} \neq \emptyset$).  For completion, $\min(\emptyset) = n+1$.
	
    As a consequence of~\Cref{cor:behaviorOnO}, no agent $i\leq K$ receives small goods in both the computed allocations, so $j_1, j_2 > K$ and hence $h > K$.
	
	We first show property (1). Let $i \in \agents \setminus N'$. If $i \le K$, then, as the coupled process applies the same operations on agents $i \le K$, we have $v_{\sigma(i)}(\tilde O^{\tilde T}_{\sigma(i)}) = v_i(\tilde O^{\tilde T}_i) = v_i(O^{\tilde T}_i)$. If $K < i$, by definition of $h$, neither $i$ (in $O$) nor $\sigma(i)$ (in $\tilde O$) did give or receive a good, so $v_i(O^{\tilde T}_i) = b_i = \tilde b^\sigma_i = v_{\sigma(i)}(\tilde O^{\tilde T}_{\sigma(i)})$.
	
	Now, property (2) is immediate since both allocations have the same total utilitarian social welfare, i.e., the utilitarian social welfare for agents in $N'$ is the same in both allocations. Let us denote it by $U$.
	
	We now show property (3) and distinguish the cases $h=j_1$, $h=j_2 < H$, and $h=H$. If $h = n+1$, we have nothing to show, so we assume that $h < n+1$.
    
    \paragraph{Case $h = j_1$.} 
    Consider $i \geq j_1$. If $i$ has a small good, then $v_i(O^{\tilde T}_i) \leq \min_{j \in \agents} v_j(O^{\tilde T}_j)+v$ (because of the greedy allocation). If not, then, since $j_1 > K$, we see $v_i(O^{\tilde T}_i) = b_i \leq b_{j_1} \leq v_{j_1}(O^{\tilde T}_{j_1}) \leq \min_{j \in \agents} v_j(O^{\tilde T}_j)+v$.  	

    \paragraph{Case $h = j_2 < H$.} 
    If there is an agent $i \in N'$ with maximum valuation that has at least one small good in her bundle, then $v_i(O^{\tilde{T}}_i) \geq v_j(O^{\tilde{T}}_j) \geq v_i(O^{\tilde{T}}_i) -v$ for every $j_2 \leq j$. Property 3 is satisfied. 
	 
    Otherwise, $j_2$ has maximum valuation among the agents from $N'$ since $b$ is non-increasing. Hence, we have the upper bound 
	\[
	U \leq  b_{j_2}\cdot (n -j_2 +1) \enspace . 
	\]
	On the other hand, by definition of $j_2$ (and since $j_2 = h < n+1$), agent $\sigma(j_2)$ is receiving a small good in $\tilde{O}^{\tilde{T}}$. Hence, every agent has valuation at least $v_{\sigma(j_2)}(\tilde{O}^{\tilde{T}}_{\sigma(j_2)}) - v$ (with $j_2$ having a strictly larger one). Thus 
	\[
	U >  \left(v_{\sigma(j_2)}(\tilde{O}^{\tilde{T}}_{\sigma(j_2)}) - v \right) \cdot (n -j_2 +1)  \geq \tilde{b}^\sigma_{j_2} \cdot (n -j_2 +1) \enspace ,
	\]
	where the last inequality holds since $\tilde{b}^\sigma_{j_2}$ is the number of big goods of $\sigma(j_2)$ in $\tilde O^{\tilde T}$ and $\sigma(j_2)$ has at least one small good. This implies $b_{j_2} > \tilde{b}^\sigma_{j_2}$, but since $j_2< H$, we must have $b_{j_2}= \tilde{b}^\sigma_{j_2}$ -- a contradiction.

    \paragraph{Case $h = H$.} 
    Again, if a maximum valuation agent from $N'$ receives a small good, then property 3 follows. 
    
    Otherwise, agent $H$ has maximum valuation among the agents from $N'$. Hence, we can upper bound the utilitarian social welfare by $U \leq b_H\cdot(n-H+1)$. 
    
    On the other hand, we can assume the allocation $O$ Lorenz dominates\footnote{In \cite{BabaioffEF21} it is shown that we can compute in polynomial time an optimal allocation for the dichotomous submodular setting that Lorenz dominates any other allocation. Lorenz domination is defined as follows: given any two non-wasteful allocations $A$ and $A'$ and given their vectors of big goods $b$ and $d$, sorted in a non-increasing ordering, $A$ {\em Lorenz dominates} $A'$ if $\sum_{i > h}^n b_i \geq \sum_{i > h}^n d_i $ for each $h=0, \dots, n-1$.} any other non-wasteful allocation, $\tilde{O}$ included. Hence, we have $\sum_{i > h}^n b_i \geq \sum_{i > h}^n \tilde{b}^\sigma_i $ for each $h=0, \dots, n-1$. Since $H$ is the first index for which the vectors differ and $\sum_{i = 1}^n b_i = \sum_{i = 1}^n \tilde b^\sigma_i$, we see that $\sum_{i \geq H}^n b_i = \sum_{i \geq H}^n \tilde{b}^\sigma_i$ implying $\tilde{b}^\sigma_H \geq b_H$ and, in particular, $\tilde{b}^\sigma_H> b_H$ since the entries differ. Furthermore, agent $\sigma(H)$ has either strictly more or strictly less big goods in $\tilde{O}^{\tilde{T}}$ than in $O^{\tilde{T}}$. Indeed, if $\sigma(H)$ has the same number of big goods in both allocations we have $b_{\sigma(H)} = \tilde{b}_{\sigma(H)} = \tilde{b}^\sigma_{H} > b_H$, which implies $\sigma(H) < H$. Thus, by definition of $H$, $b_{\sigma(H)} = \tilde b^\sigma_{\sigma(H)}$, i.e. $\tilde b^\sigma_H = \tilde b^\sigma_{\sigma(H)}$. As $\sigma$ is injective, we have $\sigma(\sigma(H)) \neq \sigma(H)$, so $(\sigma(H), \sigma(\sigma(H))) \circ \sigma$ orders $\tilde b$ in non-increasing order and permutes less indices than $\sigma$, a contradiction to the way we choose $\sigma$. 
    
    Finally, we lower bound the valuation of any agent in $N'$ for the allocation $\tilde{O}^{\tilde{T}}$. If $\sigma(H)$ has less big goods in $\tilde{O}^{\tilde{T}}$ than in $O$, then by \Cref{cor:bigGoodsInOTilde} she receives a small good in $\tilde{O}^{\tilde{T}}$, and hence any other agent has valuation greater or equal than $v_{\sigma(H)}(\tilde{O}^{\tilde{T}}_{\sigma(H)})- v$. On the other hand, if $\sigma(H)$ has more big goods in $\tilde{O}^{\tilde{T}}$ than in $O$, then, as consequence of \Cref{cor:bigGoodsInOTilde}, there exists one agent receiving a small good who has valuation strictly greater than $\tilde{b}_H^\sigma - 1 + v \cdot \lfloor 1/v \rfloor$. Since that agent receives at least one small good, we can lower bound the valuation of any other agent by $\tilde{b}_H^\sigma - 1 + v \cdot \lfloor 1/v \rfloor -v$. 

    Hence, $ U >  (\tilde{b}_H^\sigma - 1 + v \cdot \lfloor 1/v \rfloor -v )\cdot(n-H+1) \geq (b_H + v \cdot \lfloor 1/v \rfloor -v )\cdot(n-H+1) $. This implies $b_H> b_H + v\cdot \lfloor 1/v \rfloor -v$, so $ 1 >  \lfloor 1/v \rfloor $, a contradiction.
\end{proof}

    Recall we defined $U$ as the utilitarian social welfare in $O^{\tilde{T}}$ restricted on the agents in $N'$. Denoted by $u= \frac{U}{n'}$, where $n'= \modulus{N'}$,
    The following corollary is a consequence of \Cref{prop:bigGoodsInOpt}.

\begin{corollary} \label{cor:optu2v}
    If $u < 2v$, then $O^{\tilde{T}}$ is an optimal allocation.
\end{corollary}

\begin{proof}
 Let $i \in N'$ with minimum valuation in $O^{\tilde T}$. Since $u < 2v$, we have $v_i(O^{\tilde T}_i) < 2v$, so $i$ has either nothing or one small good or one big good. Observe that the first case is not possible since $m \geq n$.
 We now show by contradiction that $\tilde O^{\tilde T}$ cannot allocate two big goods to the same agent in $\sigma(N')$. Let $j \in N'$ such that $\sigma(j)$ has (at least) two big goods in $\tilde O^{\tilde T}$. Then by \Cref{cor:bigGoodsInOTilde}, we have $k \in N'$ such that $v_{\sigma(k)}(\tilde O^{\tilde T}_{\sigma(k)}) > v_{\sigma(j)}(\tilde O^{\tilde T}_{\sigma(j)}) - 1 + v \cdot \lfloor \frac 1 v \rfloor \geq 1+v$. As $\sigma(j)$ receives a small good, every agent has valuation greater than $1$. The agents in $\sigma(N')$ have overall the same number of goods in $\tilde O^{\tilde T}$ as the agents in $N'$ in $O^{\tilde T}$, but in $O^{\tilde T}$, every agent in $N'$ has at most two goods, and $i$ has only one. So there is an agent in $\sigma(N')$ that has at most one good in $\tilde O^{\tilde T}$, and thus valuation at most $1$, a contradiction.\\
Knowing that $\tilde O^{\tilde T}$ allocates at most one good to each agent in $\sigma(N')$, we have that $O^{\tilde T}$ and $\tilde O^{\tilde T}$ have the same valuation profile, so $O^{\tilde T}$ is optimal.
\end{proof}

\begin{theorem}{theorem}
	Algorithm \ref{Algorithm} has an approximation factor of at most $\frac{24}{29} \exp\left(\frac{110}{493}\right) < 1.0345$.
\end{theorem}

\begin{proof} Recall that $O^T$ is the allocation computed by  \Balance$(\instance, O)$, $\tilde{O}^{\tilde{T}}$ is the optimal allocation computed by  \Balance$(\instance, \tilde{O})$ and $O^{\tilde{T}}$ is the allocation computed simultaneously with $\tilde{O}^{\tilde{T}}$. By \Cref{lemma:decreasingWelfare}
	$\NSW(O^T) \geq \NSW(O^{\tilde{T}})$. Moreover, $\NSW(O^*)=\NSW(\tilde{O}^{\tilde{T}})$. We see that
	\begin{dmath*}
	\frac{\NSW(O^*)}{\NSW(O^T)} \leq \frac{\NSW(\tilde{O}^{\tilde{T}})}{\NSW(O^{\tilde{T}})} \\
	=\left(  \frac{\left(\prod_{i\in N\setminus N'} v_i(\tilde{O}_{\sigma(i)}^{\tilde{T}})  \right) \left(\prod_{i\in  N'} v_i(\tilde{O}_{\sigma(i)}^{\tilde{T}}) \right)}{\left(\prod_{i\in  N \setminus N'} v_i(O_i^{\tilde{T}})\right) \left(\prod_{i\in  N'} v_i(O_i^{\tilde{T}})\right)} \right)^{\frac{1}{n}} \\
	= \left(  \frac{\prod_{i\in  N'} v_i(\tilde{O}_{\sigma(i)}^{\tilde{T}})}{\prod_{i\in  N'} v_i(O_i^{\tilde{T}})} \right)^{\frac{1}{n}}.
	\end{dmath*}
	We denote by $n'= \modulus{N'}$ and $u = \frac{1}{n'} \sum_{i \in N'} v_i(O^{\tilde T}_i)$. Then, by the AM-GM inequality,
	\[ 
	\left(  \prod_{i\in  N'} v_{\sigma(i)}(\tilde{O}_{\sigma(i)}^{\tilde{T}})\right)^{\frac{1}{n}} \leq u^{\frac{n'}{n}} \enspace .
	\]

We now provide a lower bound on $(\prod_{i \in N'} v_i(O^{\tilde T}_i))^{\frac 1 {n'}}$.

As first step towards this goal, we prove that $(\prod_{i \in N} v_i(O^{\tilde T}_i))^{\frac 1 {n'}} \geq \min\limits_{0 \leq a \leq 1} (u-av+v)^a (u-av)^{1-a}$. To decrease notational overhead, we use $u_i = v_i(O^{\tilde T}_i)$ for all $i \in N'$. 

First, if $\max_{i \in N'} u_i - \min_{i \in N'} u_i < v$, then take $i$ (resp. $j$) in $N'$ with maximum (resp. minimum) valuation, set $u_i \leftarrow \frac 1 2 (u_i+u_j+v)$ and $u_j \leftarrow \frac 1 2 (u_i+u_j-v)$. This increases $u_i$ and decreases $u_j$ without changing their arithmetic mean, so the geometric mean decreases. Thus we can assume that $\max_{i \in N'} u_i - \min_{i \in N'} u_i = v$.

Let $x = \min_{j \in N'} u_j$.

Then, if we have $i, j \in N'$ such that $u_i, u_j \in (x, x+v)$ and $u_i \leq u_j$, if $x+v - u_j < u_i - x$, we set $u_j \leftarrow x+v$ and $u_i = u_i + u_j-(x+v)$, otherwise $u_j \leftarrow u_j + x - u_i$ and $u_i \leftarrow x$. This increases $u_j$ and decreases $u_i$ without changing their arithmetic mean, so decreases the geometric mean. Thus, by repeating this transformation a finite number of times, we can assume that $\modulus{\{i \in N', x < u_i < x+v\}} \leq 1$.

 Let $f$ be the fraction of agents $i \in N'$ such that $u_i = x+v$. If, for $i \in N'$, we have $u_i \in \{x, x+v\}$, then we have (with $a = f$), $(1-a) n' x + an' (x+v) = n' u$ and thus $x = u - av$. This implies $\prod_{i \in N'} v_i(O^{\tilde T}_i)^{\frac 1 {n'}} \geq (\prod_{i \in N'} u_i)^{\frac 1 {n'}} = (u-av+v)^a(u-av)^{1-a}$. Otherwise, let $i \in N'$ such that $x < u_i < x+v$. Let $a = f + \frac 1 {n'} \frac{u_i - x} v$, so that $u_i = (1-(a-f)) n' \min_{j \in N'} u_j + (a-f) n' \max_{j \in N'} u_j$. We have $(1-a) n' x + a n' (x+v) = u$, so $x = u-av$ and we have again $(\prod_{i \in N'} v_i(O^{\tilde T}_i)^{\frac 1 {n'}} \geq (\prod_{i \in N'} u_i)^{\frac 1 {n'}} = (u-av+v)^a(u-av)^{1-a}$.

This implies 
\begin{dmath*}
\left(\prod_{i \in N} v_i(O^{\tilde T}_i)\right)^{\frac 1 {n'}} \ge \min\limits_{0 \leq a \leq 1} (u-av+v)^a (u-av)^{1-a}\\
= \exp(\min\limits_{0 \leq a \leq 1} a \ln(u-av+v)+(1-a)\ln(u-av))\\
= u\exp\left(\min\limits_{0 \leq a \leq 1} a \ln\left(1-a\frac v u + \frac v u\right)+(1-a)\ln\left(1-a \frac v u\right)\right)\enspace.
\end{dmath*}
Let $f : (0, 1)^2 \to \RR$ be a function over the open unit square, given by
\[
f(d, a) = a \ln(1+d-ad)+(1-a)\ln(1-ad).
\]
First note that $f$ is smooth on $(0, 1)^2$.

We show that for all $d \in (0, 1)$ the function $a \mapsto f(d, a)$ is convex. Towards this goal, we compute the derivatives:

\begin{dmath*}
    \frac{\partial f}{\partial a}(d, a) =  \ln\left(1+\frac d {1-ad}\right) - \frac{ad}{1+d-ad} - \frac{d-ad}{1-ad} \hspace{1cm} \text{for all }  d, a \in (0,1)
\end{dmath*}
Then
\begin{dmath*}
\frac{\partial^2 f}{\partial a^2}(d, a) = \frac{d^2}{(1-ad)^2(1+d-ad)^2} \cdot (-d^2 a^2 + d^2 a + 1-d^2) \hspace{1cm} \text{for all }  d, a \in (0, 1)
\end{dmath*}
For $a, d \in (0, 1)$, we have $d^2 < 1$ and $d^2a^2 < d^2 a$, so $0 < \frac{\partial^2 f}{\partial a^2}(d, a)$.

Let $d \in (0,1)$. By convexity, and taking the tangent at $\frac 1 2 + \frac d 6$ 
we have for all $a \in (0,1)$

\begin{dmath*}
    f(d, a) \geq  \frac{1}{6} (3+d)\ln\left(\frac{6+3d-d^2}{6-3d-d^2}\right) + \ln\left(1-\frac d 2 - \frac{d^2} 6\right) 
     \hspace{0.5cm} +\left(a-\frac 1 2 - \frac d 6 \right)\cdot \left(\ln\left(\frac{6+3d-d^2}{6-3d-d^2}\right) + 12d\cdot\frac{d^2-3}{d^4-21d^2+36}\right)
\end{dmath*}

In order to minimize an affine function, we only need to know if it is increasing or not. We define $g : (0,1) \to \RR$ by $g(d) = \ln(\frac{6+3d-d^2}{6-3d-d^2}) + 12d\frac{d^2-3}{d^4-21d^2+36}$. Then for all $d \in (0,1)$ 
\[ g'(d) = \frac{-6d^4(d^2+39)}{(d^4-21d^2+36)^2} \leq 0, \]
and we conclude 
\begin{dmath*}
    \min_{0 < a < 1} f(d, a) \geq  \frac 16 (3+d)\ln\left(\frac{6+3d-d^2}{6-3d-d^2}\right) + \ln\left(1-\frac d 2 - \frac{d^2} 6\right) 
 + \left(1-\frac 1 2 - \frac d 6\right)\cdot\left(\ln\left(\frac{6+3d-d^2}{6-3d-d^2}\right) + 12d \cdot \frac{d^2-3}{d^4-21d^2+36}\right)
 = \ln\left(1+\frac d 2 - \frac{d^2} 6\right) + \frac{2d(3-d)(d^2-3)}{d^4-21d^2+36} .
\end{dmath*}

%
Note that this last function of $d$ has derivative
\[
d \mapsto \frac d {d^4-21d^2+36} \cdot (2d^6-3d^5+9d^4-117d^3+225d^2-324)
\]
which is negative over $(0,1)$. Thus, it is minimized for maximum value of $d$.

Let us now set $d = \frac v u$. If $\frac 1 2 < d$, by \Cref{cor:optu2v}, $O^{\tilde T}$ is optimal, so we can assume $d \le \frac 1 2$.
Hence we have 
\[\begin{array}{rcl}
\min_{0 < a < 1} f(d, a)&\geq& \ln(\frac{29}{24}) - \frac{110}{493} \enspace . \end{array}\]

This implies $(\prod_{i \in N'} v_i(O^{\tilde T}_i))^{\frac 1 {n'}} \geq u \cdot (\frac{29}{24} \exp(-\frac{110}{493}))$, and, thus,

\begin{dmath*}
    \frac{\NSW(O^*)}{\NSW(O^T)} \leq \left(\frac{24}{29} \exp\left(\frac{110}{493}\right)\right)^{\frac{n'} n}
\leq\frac{24}{29} \exp\left(\frac{110}{493}\right) < 1.0345.    
\end{dmath*}

\end{proof}

Observe that \Cref{ex:OPTvsOPTstar} provides a lower bound for the approximation ratio of Algorithm \ref{Algorithm} of $6/\sqrt{35} \approx  1.01418$. It is an interesting open problem whether the approximation guarantee for the algorithm can be made tight.

\section{\classNP-Hardness when $p\geq 3$}\label{sec:NPhard}


In this section we almost complement our positive results on polynomial-time NSW optimization. In particular, we show: 

\begin{theorem}\label{NP-hard}
It is \classNP-hard to compute an allocation with optimal NSW for 2-value instances, for any constant coprime integers $q>p\geq 3$.
\end{theorem}

We provide a reduction from \emph{Exact-$p$-Dimensional-Matching (Ex-$p$-DM)}: 
Given a graph $G$ consisting of $p$ disjoint vertex sets $V_1, \dots, V_p$, each of size $n$, and a set $E\subseteq V_1\times\dots\times V_p$ of $m$ edges, it is \classNP-hard to decide whether there exists a $p$-dimensional perfect matching in $G$ or not. Note that for $p=3$ the problem is Ex-3-DM and thus \classNP-hard. \classNP-hardness for $p>3$ follows by simply copying the third set of vertices in the Ex-3-DM instance $p-3$ times, thereby also extending the edges to the new vertex sets.

\paragraph{Transformation:} 
There is one good for each vertex of $G$, call them \emph{vertex goods}. Additionally, there are $q(m-n)$ \emph{dummy goods}. For each edge of $G$, there is one agent who values the $p$ incident vertex goods $1$ and all other goods $\frac{p}{q}$.

\begin{lemma}\label{lemma:NPh1}
    If $G$ has a perfect matching, then there is an allocation $\mathcal{A}$ of goods with $\NSW(\mathcal{A})=p$.
\end{lemma}

\begin{proof}
Suppose there exists a perfect matching in $G$. We allocate the goods as follows: Give each agent corresponding to a matching edge all $p$ incident vertex goods. Now there are $m-n$ agents left. Give each of them $q$ dummy goods. As each agent has valuation $p$, the NSW is $p$ as well.
\end{proof}

\begin{lemma}\label{lemma:NPh2}
    If $G$ has no perfect matching, then for every allocation $\mathcal{A}$ of goods,  $\NSW(\mathcal{A})<p$.
\end{lemma}

\begin{proof}
Suppose there is an allocation $\mathcal{A}=(A_1,\dots,A_m)$ of goods with $\NSW(\mathcal{A})\geq p$. We show that in this case there must be a perfect matching in $G$.
First, observe that if we allocate each good to an agent with maximal value for it, we obtain an upper bound on the average sum social welfare of $\mathcal{A}$, i.e.\ $\tfrac{1}{m}\sum_i v_i(A_i) \leq \tfrac{1}{m}(pn + q(m-n)\cdot \tfrac{p}{q}) = p$. 
Applying the AM-GM inequality gives us also $\NSW(\mathcal{A})=\left( \prod_i v_i(A_i) \right)^{1/m}\leq p$, and, in particular, $\NSW(\mathcal{A}) = p$ iff $v_i(A_i) = p$ for all agents $i$.
Hence each agents valuation must be $p$ in $\mathcal{A}$ and each vertex good must be allocated to an incident agent. The next claim allows to conclude that there are only two types of agents in $\mathcal{A}$:

\begin{claim*}
    If an agent $i$ has valuation $v_i(A_i) = p$, then she either gets her $p$ incident vertex goods or $q$ other goods.
\end{claim*}

    We show that $(p,0)$ and $(0,q)$ are the only integral solutions $(i,j)$ of the equation $p = i + j\tfrac{p}{q}$, where $i,j\geq 0$. Clearly, every solution different from the above must satisfy $0<i<p$. Assume for contradiction that there exists such a solution. Then it must hold $(p-i)q = jp$. Since $p$ and $q$ are coprime, $p-i$ must be a multiple of $p$ and thus $i\in\{0,p\}$, a contradiction. This concludes the proof of the claim.

Let $b$ be the number of agents receiving their $p$ incident vertex goods in $\mathcal{A}$, and $m-b$ the number of agents receiving $q$ other goods. Since each vertex good must be allocated to an incident agent, $bp=np$ and thus $b=n$. Hence there must be $n$ agents receiving their $p$ incident vertex goods, which implies that there is a perfect matching in $G$.
\end{proof}

\Cref{lemma:NPh1} and \Cref{lemma:NPh2} yield the proof of Theorem \ref{NP-hard}.

\section{\classAPX-Hardness when $p/q = 4/5$}\label{sec:apxhardness}
	
	The main result in this section is \classAPX-hardness for general 2-value instances.
    
    \begin{theorem}\label{theo:approxHardness}
	It is \classNP-hard to approximate the maximum \NSW\ for 2-value instances to within a factor better than $1.000015$.
    \end{theorem}
	
	We provide a gap-preserving reduction from \emph{Gap-4D-Matching} with almost perfect completeness. The instance transformation is similar to the \classNP-hardness construction, but with a slightly different number of dummy goods.
	
	\begin{lemma}[see \cite{HazanSS03}, section 3]\label{lemma:gap4dMatching}
		For all sufficiently small $\epsilon>0$, the following holds: Given a graph $G$ consisting of disjoint vertex sets $V_1, V_2, V_3, V_4$, each of size $n$, and a set $E\subseteq V_1\times V_2\times V_3\times V_4$ of $m=3n$ edges, it is \classNP-hard to distinguish whether there exists a matching in $G$ with $(1-\epsilon)n$ edges or each matching has less than $(\tfrac{53}{54}+\epsilon)n$ edges.
	\end{lemma}
	
	Before proceeding to the proof of our main theorem, let us note that in \cite{HazanSS03} the authors do not explicitly state that the number of edges is $m=3n$ in their construction. We give a short overview of their reduction and argue that this is indeed true.
	
    Starting point is the problem \emph{Max-E3-Lin(2)}. The input of this problem is a set $\Phi$ of equations of the form $\varphi : x+y+z = t \; (\text{mod } 2)$, where $t\in\{0,1\}$ and $x,y,z$ are boolean variables from a set $X$. The \classNP-hard gap problem is to decide whether there exists an assignment to the variables satisfying a $(1-\epsilon)$-fraction of equations or each assignment satisfies less than a $(\tfrac{1}{2}+\epsilon)$-fraction, see \cite{Hastad01}.
		
    Now we describe the transformation to 4D-Matching.
	Let $d_x$ be the number of occurrences of variable $x\in X$ in the input. It can be assumed that each equation contains three different variables and hence $3|\Phi| = \sum_{x\in X} d_x$. Furthermore, it can be assumed that each variable occurs equally often in the first, second, and third position of an equation.
		
	The construction is based on two steps: In the first step, the so called \emph{consistency-gadget} is constructed, a 3-partite 3-uniform hyper-graph $G=(V_1,V_2,V_3, F)$, where $V_1, V_2, V_3$ are equally sized vertex partitions and $F\subseteq V_1\times V_2\times V_3$ is a set of degree-3 hyper-edges. In the second step, the so called \emph{gap-gadget} is introduced into $G$. In this step, the edges of $G$ are connected to a fourth vertex set $V_4$ to build the resulting 4D-Matching instance $G'=(V_1,V_2,V_3,V_4,F')$ with edge set $F'\subseteq V_1\times V_2\times V_3\times V_4$. Now we describe these two steps.
		
	Consistency-gadget:
	For each variable $x\in X$, a 3-regular, edge-3-colorable, bipartite disperser graph $\hat{G}_x=(U_x^0\cup U_x^1,E_x^1\cup E_x^2\cup E_x^3)$ is created, where $U_x^0$ and $U_x^1$ are the vertex partitions, each of size $\tfrac{4}{3} d_x$, and the $E_x^c$ are the three color classes of edges, also each of size $\tfrac{4}{3}d_x$. The two partitions correspond to the two possible assignments to variable $x$, and dispersing properties of $\hat{G}_x$ imply that in every large independent set of $\hat{G}_x$, a large set of vertices stems from only one partition. Now all $\hat{G}_x$ are converted to their dual graphs $G_x = (V_x=V_x^1\cup V_x^2\cup V_x^3, F_x=F_{x=0}\cup F_{x=1})$, that means each edge of $E_x^c$ is now a degree-2 vertex in $V_x^c$ and each vertex in $U_x^0$ (resp.\ $U_x^1$) is now a degree-3 hyper-edge in $F_{x=0}$ (resp.\ $F_{x=1}$). Note that every independent set in $\hat{G}_x$ corresponds to a matching in $G_x$ and vice versa. Now define the graph $G=(V_1,V_2,V_3, F)$ with $V_c=\bigcup_{x\in X} V_x^c$ and $F=\bigcup_{x\in X} F_x$ to be the union of all $G_x$. Note that $|V_c|=\sum_{x\in X} \tfrac{4}{3} d_x = 4|\Phi|$.
		
	Gap-gadget:
	Now the graph $G$ is extended to the final 4D-Matching instance $G'$. Therefore, the edges of $G$ are amended with a vertex from a fourth vertex set $V_4$. $V_4$ contains for each equation $\varphi\in \Phi$ four vertices $v_\varphi^1, v_\varphi^2, v_\varphi^3, v_\varphi^4$. Note that $|V_4| = 4|\Phi|$, and hence $|V_1|=|V_2|=|V_3|=|V_4| =: n$. 
	The connection scheme for an equation $\varphi : x+y+z = 0 \; (\text{mod } 2)$ is as follows (the case with a right-hand-side of 1 is similar):
	\begin{enumerate}
	    \item Pick two edges $e,f\in F_{x=0}$ and add $v_\varphi^1$ to $e$ and $v_\varphi^2$ to $f$.
	    \item Pick two edges $e,f\in F_{x=1}$ and add $v_\varphi^3$ to $e$ and $v_\varphi^4$ to $f$.
	    \item Pick one edge $e\in F_{y=0}$, duplicate it and add $v_\varphi^1$ to $e$ and $v_\varphi^3$ to the copy.
	    \item Pick one edge $e\in F_{y=1}$, duplicate it and add $v_\varphi^2$ to $e$ and $v_\varphi^4$ to the copy.
	    \item Pick one edge $e\in F_{z=0}$, duplicate it and add $v_\varphi^1$ to $e$ and $v_\varphi^4$ to the copy.
	    \item Pick one edge $e\in F_{z=1}$, duplicate it and add $v_\varphi^2$ to $e$ and $v_\varphi^3$ to the copy.
	\end{enumerate}
	Each edge of $G$ will be picked only once. Now consider a variable $x\in X$ after this process. The number of picked edges from $F_{x=0}$ is $\tfrac{1}{3}d_x\cdot 2 + \tfrac{2}{3}d_x \cdot 1 = \tfrac{4}{3}d_x = |F_{x=0}|$, as $1/3$ of the occurrences of $x$ are in the first position and $2/3$ are in the second or third position. The number of duplicated edges is $\tfrac{2}{3}d_x$. The same holds for $F_{x=1}$. Hence the overall number of edges is $m := \sum_{x\in X} 2(\tfrac{4}{3} d_x + \tfrac{2}{3} d_x) = 12|\Phi| = 3n$.

	The hard gap of the matching instance follows from the hard gap of Max-E3-Lin(2) and dispersing properties of the graphs $\hat{G}_x$.
		
	\begin{proof}[Proof of Theorem~\ref{theo:approxHardness}]
	    Let us now proceed to the proof of our main theorem.
		\paragraph{Transformation:}
		There is one good for each vertex, call them \emph{vertex goods}. Additionally, there are $5(m-(1-\epsilon)n)$ \emph{dummy goods}.
		For each edge, there is one agent who values the four incident vertex goods $1$ and all other goods $\tfrac{4}{5}$.
		
		\paragraph{Completeness:}
		Suppose there exists a matching of size $(1-\epsilon)n$. Then there is an allocation $\mathcal A$ with $\text{NSW}(\mathcal A) = 4$:
		Give each agent corresponding to a matching edge all four incident vertex goods. Now there are $m-(1-\epsilon)n$ agents left. Give each of them $5$ dummy goods. As each agent has valuation $4$, the NSW of this allocation is also $4$.
		
		\paragraph{Soundness:}
		Suppose every matching is smaller than $(\tfrac{53}{54}+\epsilon)n$. We show that in this case the NSW of any allocation is upper bounded by a constant strictly smaller than $4$.
		
		First, observe that in every allocation, each agent valuation is of the form $u_{ij} = i+\tfrac{4}{5}j$ with integers $i\in \{0,\dots,4\}$ and $j\geq 0$, meaning that the agent receives $i$ goods she values at $1$ and $j$ goods she values at $\tfrac{4}{5}$. We call the pairs $(i,j)$ \emph{valuation types}.
		Using the next two lemmas, we show that the number of different valuation types which may occur in Nash optimal allocations is limited by a constant.
		
		\begin{restatable}{lemma}{utilityBoundOne}\label{lemma:utilityBound1}
			In each allocation there is an agent with valuation at most $4$.
		\end{restatable}
		
		\begin{proof}
			Suppose not. Then every agent has valuation at least $4.2$, as this is the smallest possible valuation strictly greater than $4$. Hence the sum of agents valuations is at least $4.2m$. However, if each good is allocated to an agent with maximal value for it, we get a total sum of valuations of $4n + 5(m-(1-\epsilon)n)\cdot \tfrac{4}{5} = 4m + 4\epsilon n$, which is strictly smaller than $4.2m$ for sufficiently small $\epsilon$ (note that $n=m/3$). This is a contradiction.
		\end{proof}
		
		\begin{restatable}{lemma}{utilityBoundTwo}\label{lemma:utilityBound2}
			In a \NSW-optimal allocation, every agent has valuation at most $4.8$.
		\end{restatable}
		
		\begin{proof}
			Suppose there is an agent $k$ with valuation strictly greater than $4.8$ in a \NSW-optimal allocation. Then $k$ has at least two small valued goods. According to Lemma \ref{lemma:utilityBound1} there is an agent $k'$ with valuation at most $4$. Giving one of the small goods of $k$ to $k'$ strictly increases the NSW which contradicts optimality.
	    \end{proof}

		Now consider an optimal allocation $O^*$. From Lemma \ref{lemma:utilityBound2} it follows that the set of valuation types $(i,j)$ occurring in $O^*$ is a (strict) subset of $I := \{0,\ldots,4\} \times \{0,\ldots,6\}$, as otherwise there would be an agent with valuation more than $4.8$.
		This allows us to construct an LP with $|I|+1=36$ variables, whose optimal value yields an upper bound on $\text{NSW}(O^*)$:
		
		Let $x_{ij}$ be the fraction of agents with valuation $u_{ij}$ in $O^*$. Furthermore, let $\alpha$ be the fraction of small allocated vertex goods in $O^*$. Then the following must hold:
		\begin{align}
		\sum_{ij\in I} x_{ij} &= 1 \tag{no other valuation types occur}\\
		\sum_{j=0\ldots 6} x_{4j} m &\leq \left( \tfrac{53}{54} + \epsilon \right)n \tag{otherwise there would be a `large' matching}\\
		\sum_{ij\in I} ix_{ij}m &\leq \left( 1-\alpha \right) \hspace{-6mm} \underbrace{4n}_{\text{\# vertex goods}} \tag{\# big valued goods not exceeded}\\
		\sum_{ij\in I} jx_{ij}m &\leq \underbrace{5\left(m-(1-\epsilon)n\right)}_{\text{\# dummy goods}} + 4\alpha n \tag{\# small valued goods not exceeded}
		\end{align}
		
		Now observe that $\text{NSW}(O^*) = \prod_{ij\in I} u_{ij}^{x_{ij}}$, hence by applying the logarithm and substituting $m=3n$ in the above inequalities we obtain an LP whose optimal solution yields an upper bound on $\log \text{NSW}(O^*)$:
		
		\begin{align*}
		\max_{x} &\; \sum_{ij\in I} x_{ij} \ln u_{ij} \quad \text{s.t.} \\
		&\sum_{ij\in I} x_{ij} = 1
		&&\sum_{j=0\ldots 6} x_{4j} \leq \tfrac{1}{3}\left( \tfrac{53}{54} + \epsilon \right) \\
		&\sum_{ij\in I} ix_{ij} \leq \tfrac{4}{3}\left( 1-\alpha \right)
		&&\;\; \sum_{ij\in I} jx_{ij} \leq \tfrac{1}{3}\left( 10 + 5\epsilon \right) + \tfrac{4}{3}\alpha \\
		&\; x_{ij} \geq 0, \quad \forall ij\in I \\
		&\; 0\leq \alpha \leq 1
	    \end{align*}
		
		Solving the LP for $\epsilon=0$ shows an optimum at

		\[\alpha{=}0,\;\;\; x_{40}{=}\tfrac{53}{162},\;\;\;  x_{14}{=}x_{31}{=}\tfrac{1}{162},\;\;\;\]
		\[x_{05}{=}\tfrac{107}{162},\;\;\; x_{ij}{=}0 \;\;\; \text{otherwise.} \]
		
		Hence $\text{NSW}(O^*)\leq (4.2\cdot 3.8)^{1/162}\cdot 4^{160/162} + \delta$, where $\delta\to 0$ as $\epsilon\to 0$.\\
		
		From our completeness and soundness arguments it follows that for all sufficiently small $\delta>0$, it is \classNP-hard to distinguish between problem instances where the maximum NSW is at least 4 and instances where the maximum NSW is at most $(4.2\cdot 3.8)^{1/162}\cdot 4^{160/162} + \delta$.
		Therefore, it is \classNP-hard to approximate the maximum Nash social welfare with a factor smaller than $\frac{4}{(4.2\cdot 3.8)^{1/162}\cdot 4^{160/162}} = 1.00001545\ldots$
	\end{proof}

\bibliographystyle{abbrv}
\bibliography{main}
\end{document}